\documentclass[journal,twocolumn,twoside]{IEEEtran} 

\usepackage[T1]{fontenc}

\usepackage{graphicx}
\usepackage[export]{adjustbox}
\usepackage[dvipsnames]{xcolor}
\usepackage{bbm}
\usepackage{makecell}
\usepackage{booktabs}
\usepackage{threeparttable}
\usepackage{comment}
\usepackage{epstopdf}

\usepackage{subfigure}
\usepackage{stfloats}
\ifCLASSINFOpdf
\else
\fi

\def\nI{n_{\mathrm{I}}}

\def\uo{\omega_0} 
\def\un{\omega_n} 

\def\SINR{\mathsf{SINR}}

\def\Pt{P_{\mathrm{t}}}

\def\Gn{G_{n}}
\def\Go{G_{0}}
\def\Gtn{G_{\mathrm{t},n}}

\def\Gto{G_{\mathrm{t},0}}

\def\Gr{G_{\mathrm{r}}}

\def\PsiG{\Psi}

\def\Avis{\mathcal{A}_{\mathrm{vis}}}

\def\BPPg{\Phi}
\def\ag{a}
\def\Ag{\mathcal{A}}
\def\Agvis{\Ag_{\mathrm{vis}}}

\def\Pcov{P_{\mathrm{cov}}}

\def\rmin{r_{\mathrm{min}}(\phi)}
\def\rmax{r_{\mathrm{max}}(\phi)}
\def\romax{r_{\mathrm{vis,max}}}

\def\PsuccV{p_{\mathrm{vis}}}
\def\PsuccI{p_{\mathrm{int}}}

\def\fc{f_{\mathrm{c}}}

\def\phiinv{\phi_{\mathrm{inv}}}

\def\E{\mathbb{E}}
\def\P{\mathbb{P}}

\def\re{r_{\mathrm{E}}}

\usepackage{amssymb}
\usepackage{amsmath}
\usepackage[cmintegrals]{newtxmath}
\usepackage{stackengine}

\usepackage{algorithmic}
\usepackage{algorithm}

\usepackage{amsthm}
\usepackage{xpatch}

\newtheorem{thm}{Theorem}
\newtheorem{lem}{Lemma}

\newtheorem{cor}{Corollary}
\newtheorem{rem}{Remark}

\xpatchcmd{\proof}{\hskip\labelsep}{\hskip5\labelsep}{}{}

\usepackage{mathtools}
\makeatletter
\newcommand{\vast}{\bBigg@{4}}
\newcommand{\Vast}{\bBigg@{5}}

\usepackage{array}

\begin{document}
\title{Modeling and Analysis of GEO Satellite Networks}

\author{Dong-Hyun Jung, Hongjae Nam, Junil Choi, and David J. Love\\
\thanks{D.-H. Jung is with the Satellite Communication Research Division, Electronics and Telecommunications Research Institute, Daejeon, 34129, South Korea (e-mail: dhjung@etri.re.kr).}
\thanks{H. Nam and D. J. Love are with the School of Electrical and Computer Engineering, Purdue University, West Lafayette, IN 47907 USA (e-mail: nam86@purdue.edu; djlove@purdue.edu).}
\thanks{J. Choi is with the School of Electrical Engineering, KAIST, Daejeon, 34141, South Korea (e-mail: junil@kaist.ac.kr).}

\vspace{-0.7cm}
}
\maketitle

\begin{abstract}
The extensive coverage offered by satellites makes them effective in enhancing service continuity for users on dynamic airborne and maritime platforms, such as airplanes and ships. In particular, geosynchronous Earth
orbit (GEO) satellites ensure stable connectivity for terrestrial users due to their stationary characteristics when observed from Earth. This paper introduces a novel approach to model and analyze GEO satellite networks using stochastic geometry. 
We model the distribution of GEO satellites in the geostationary orbit according to a binomial point process (BPP) and examine satellite visibility depending on the terminal's latitude. Then, we identify potential distribution cases for GEO satellites and derive case probabilities based on the properties of the BPP. We also obtain the distance distributions between the terminal and GEO satellites and derive the coverage probability of the network. We further approximate the derived expressions using the Poisson limit theorem. Monte Carlo simulations are performed to validate the analytical findings, demonstrating a strong alignment between the analyses and simulations. The simplified analytical results can be used to estimate the coverage performance of GEO satellite networks by effectively modeling the positions of GEO satellites.

\textbf{\emph{Index terms}} --- Satellite communications, coverage analysis, stochastic geometry, GEO satellite networks.
\\
\end{abstract}

\IEEEpeerreviewmaketitle

\vspace{-0.6cm}

\section{Introduction}\label{sec:Intro}
\IEEEPARstart{S}{atellite}
communications have recently been utilized to offer worldwide internet services by taking advantage of their extensive coverage.
In this regard, the 3rd Generation Partnership Project (3GPP) has been working toward the integration of terrestrial networks (TNs) and non-terrestrial networks (NTNs) since Release 15 [\ref{Ref:3GPP_38.811}], [\ref{Ref:3GPP_38.821}].
The utilization of non-terrestrial entities, such as geosynchronous Earth orbit (GEO) satellites, low Earth orbit (LEO) satellites, and high-altitude platforms, presents an opportunity to extend communication services beyond terrestrial boundaries. With this advancement, aerial users like drones, airplanes, and vehicles involved in urban air mobility could benefit from enhanced connectivity. 
To facilitate the integration between the TNs and NTNs, 3GPP has been addressing adding features to the standard to support NTNs with existing TNs~[\ref{Ref:3GPP_38.821}].

The different altitudes and stationary nature between GEO and LEO satellites result in distinct characteristics in terms of communication services and orbital configurations.
In general, LEO satellites could offer greater throughput due to the lower path loss compared to that of GEO satellites, whereas the high-speed movement of LEO satellites leads to frequent inter-satellite handovers. 
On the contrary, GEO satellites, being viewed as stationary, could maintain stable connections with ground users at the cost of relatively lower throughput.
In the design of LEO constellations, Walker Delta constellations with various inclinations are utilized to achieve uniform coverage near the equator, while Walker Star constellations are employed to provide services to the polar regions [\ref{Ref:Su-22}]. In contrast to the LEO satellites, GEO satellites are positioned exclusively in the geostationary orbit on the equatorial plane to maintain a stationary view from Earth.

\subsection{Related Works}\label{sec:Intro_rel_work}

The system-level performance of LEO satellite networks has been recently evaluated using stochastic geometry where the positions of LEO satellites are effectively modeled by spatial point processes. Binomial point processes (BPPs) have been widely used to model the distributions of LEO satellites because the total number of LEO satellites is deterministic [\ref{Ref:Okati1}]-[\ref{Ref:Jung1}].
The initial work [\ref{Ref:Okati1}] provided BPP-based coverage and rate analyses and showed that the BPP satisfactorily models deterministic Walker constellations.
The distance distribution between the nearest points on different concentric spheres was obtained in [\ref{Ref:Talgat1}]. With this distribution, the coverage probability of LEO satellite communication networks was derived in [\ref{Ref:Talgat2}], specifically examining the role of gateways as relays between the satellites and users. 
The ergodic capacity and coverage probability of cluster-based LEO satellite networks were evaluated in [\ref{Ref:Jung1}] considering two different types of satellite clusters.
However, the BPP-based analytical results include highly complex terms, which are not fairly tractable to evaluate network performance.

Instead of the BPPs, Poisson point processes (PPPs) could be used to approximately model the LEO satellite constellations using the Poisson limit theorem when a large number of satellites exists [\ref{Ref:Jung2}]-[\ref{Ref:Park}]. Both the BPP- and PPP-based performance analyses were carried out in [\ref{Ref:Jung2}] under the shadowed-Rician fading in terms of the outage probability and system throughput.
The downlink coverage probability of LEO satellite networks was derived in [\ref{Ref:Al-Hourani1}] considering a recent satellite-to-ground path loss model and an elevation angle-dependent line-of-sight (LOS) probability. The altitude of satellite constellations was optimized in [\ref{Ref:Al-Hourani2}] to maximize the downlink coverage probability.
In [\ref{Ref:Okati2}], a non-homogenous PPP was used to model the non-uniform distribution of LEO satellites across different latitudes. More tractable results for the coverage probability were provided in [\ref{Ref:Park}] where the density of LEO satellites was also optimized. 

The link-level performance of GEO satellite systems has been also analyzed from various perspectives [\ref{Ref:Abdu}]-[\ref{Ref:Fortes}]. The earlier studies [\ref{Ref:Abdu}], [\ref{Ref:Abdu2}] introduced a flexible resource allocation design for GEO satellite systems to maximize spectrum utilization. The primary focus was on minimizing the number of frequency carriers and transmit power required to meet the demands of multi-beam scenarios. The coverage of GEO satellite systems was enhanced in [\ref{Ref:Jia}] using weighted cooperative spectrum sensing among multiple GEO satellites via inter-satellite links. In [\ref{Ref:Khan}], reflecting intelligent surfaces were integrated into a downlink GEO scenario where the joint power allocation and phase shift design problem was efficiently solved. The interference analyses between a GEO satellite and a LEO satellite were conducted in [\ref{Ref:CS.Park}] to assess the effectiveness of an exclusive angle strategy in mitigating in-line interference. In [\ref{Ref:Fortes}], the interference analysis between a single GEO satellite and multiple LEO satellites was introduced based on the probability density functions (PDFs) of the LEO satellites' positions. Although the above works [\ref{Ref:Okati1}]-[\ref{Ref:Fortes}] have successfully investigated diverse aspects of satellite communication systems, they have primarily focused on the satellites in non-geostationary orbits or a single GEO satellite.

\subsection{Motivation and Contributions}\label{sec:Intro_motiv_cont}
GEO satellites appear stationary when observed from Earth because they move in the same direction as the Earth's rotation, and their orbits are positioned on the equatorial plane, i.e., the inclination of zero degrees.\footnote{ While GEO satellites could have inclinations greater than or equal to zero degrees, our paper focuses specifically on geostationary orbit satellites for communication purposes, which are characterized by inclinations close to zero degrees.
}  
Due to this inherent property of the geostationary orbit, modeling the positions of multiple GEO satellites is different from the techniques for modeling multiple LEO satellites but has not been addressed before to the best of our knowledge.                                               
To fill this gap, we investigate the fundamental framework of GEO satellite networks by leveraging the orbital characteristics of the GEO satellites.
The key contributions of this paper are summarized as follows.
\begin{itemize}
    \item \textbf{Modeling and analysis of GEO satellite networks:} The orbital characteristics of the GEO satellites lead to dissimilar dimensional geometry. The terminals are on the surface of Earth (3D distribution), while the GEO satellites are located on the equatorial plane (2D distribution). This makes the critical differences from the works [\ref{Ref:Okati1}]-[\ref{Ref:Park}] that considered LEO constellations.
    Hence, terminals at different latitudes experience unequal satellite visibility, resulting in performance gaps. 
    Considering these characteristics, we provide a novel approach to model and analyze GEO satellite networks based on stochastic geometry.
    Specifically, we distribute GEO satellites in the geostationary orbit according to a BPP and then analyze satellite visibility depending on the terminal's latitude. We identify the possible distribution cases for the GEO satellites and derive the probabilities of these cases based on the properties of the BPP. We also obtain the distance distributions between the terminal and GEO satellites and then derive the coverage probability using these distributions.
    \item \textbf{Poisson limit theorem-based approximation:} We approximate the satellite distribution as a PPP using the Poisson limit theorem. With this approach, the derived satellite-visible probability, distance distributions, and coverage probability are further simplified. Using the two-line element\footnote{A two-line element set is a data format encoding the list of orbital elements of a satellite at a given epoch time, which is publicly provided by the North American Aerospace Defense Command. Based on the two-line element set, the position and velocity of the satellite could be predicted by using a simplified general perturbation model, e.g., SGP4 [\ref{Ref:SGP4}].} dataset of the currently active GEO satellites, we compare the average number of visible satellites between the actual distribution and the BPP model. Additionally, we explore the performance gap between the BPP- and PPP-based satellite distributions, thereby identifying the conditions under which the Poisson limit theorem applies to the distribution of GEO satellites. 
\end{itemize}

The rest of this paper is organized as follows.
In Section~\ref{sec:System_model}, the network model for a GEO satellite communication network is described. In Section~\ref{sec:Math}, the orbit visibility and distance distributions are analyzed. In Section~\ref{sec:Coverage}, the analytical expressions of the coverage probability are derived.
In Section~\ref{sec:Sim_results}, simulation results are provided to validate our analysis, and conclusions are drawn in Section \ref{sec:Conclusions}.

\emph{Notation:}
$\P[\cdot]$ indicates the probability measure, and $\E[\cdot]$ denotes the expectation operator.
The complement of a set $\mathcal{X}$ is $\mathcal{X}^{\mathrm{c}}$.
$\binom{n}{k}$ denotes the binomial coefficient. $\mathrm{Bin}(n,p)$ denotes the binomial distribution with the number of trials $n$ and the success probability $p$.
The cumulative distribution function (CDF) and the PDF of a random variable $X$ are $F_X(x)$ and $f_x(x)$, respectively.
$\Gamma(\cdot)$ is the Gamma function, and the Pochhammer symbol is defined as $(x)_n=\Gamma(x+n)/\Gamma(x)$.
The Euclidean norm of a vector $\mathbf{x}$ is $||\mathbf{x}||$.
The Lebesgue measure of a region $\mathcal{X}$ is $|\mathcal{X}|$, which represents the volume of $\mathcal{X}$.
The unit step function is $u(\cdot)$, and 
the Dirac delta function is $\delta(\cdot)$.

\begin{figure}
\begin{center}
\includegraphics[width=\columnwidth]{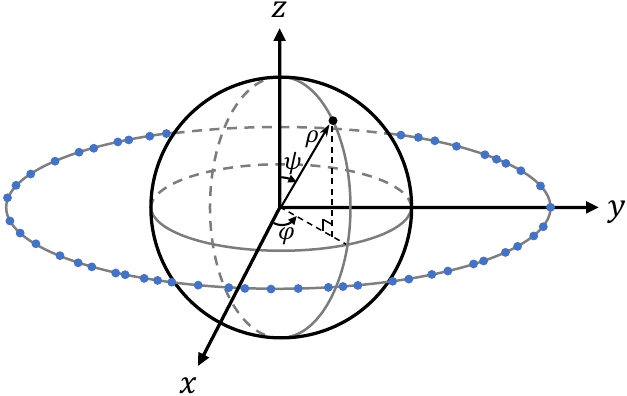}
\end{center}
\setlength\abovecaptionskip{.25ex plus .125ex minus .125ex}
\setlength\belowcaptionskip{.25ex plus .125ex minus .125ex}
\caption{An example of the GEO satellite distribution. The $xy-$plane represents the equatorial plane. The blue spheres and the gray outer circle indicate the GEO satellites and the geostationary orbit, respectively. }
\label{Fig:GEONet}
\end{figure}

\begin{figure*}[!t]
    \begin{threeparttable}
    \centering
    \renewcommand*{\arraystretch}{1.4}
    \begin{tabular}{>{\centering\arraybackslash}m{2.8cm}>{\centering\arraybackslash}m{7.1cm}>{\centering\arraybackslash}m{7.1cm} @{}m{0pt}@{}}
    \toprule
      \textbf{Latitude} & 
      $\qquad\qquad$\textbf{Orbit visibility} & 
      \textbf{Horizontal view} \\
     \midrule
     $\phi=0$ & 
     \includegraphics[width=0.75\columnwidth]{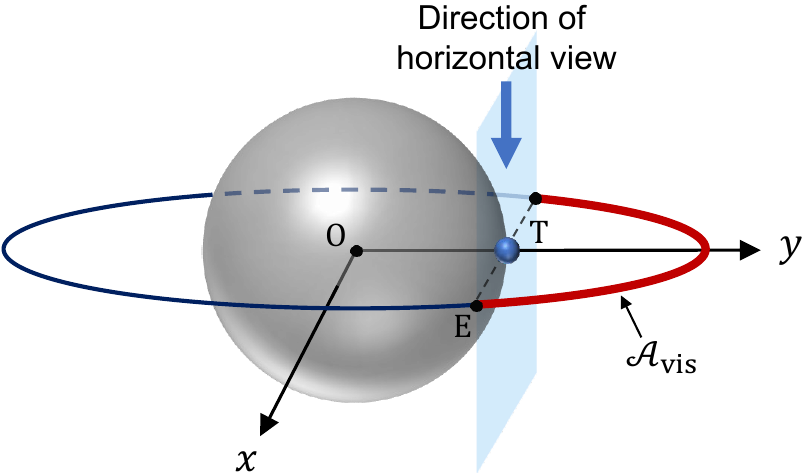} &
     \includegraphics[width=0.600\columnwidth]{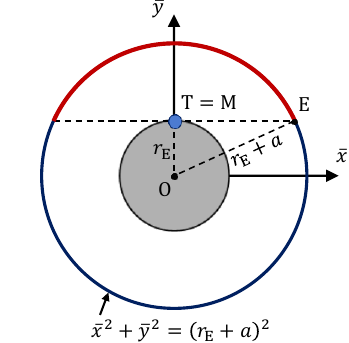} \\
     \midrule
     $0<\phi < \phiinv$ & 
     \includegraphics[width=0.75\columnwidth]{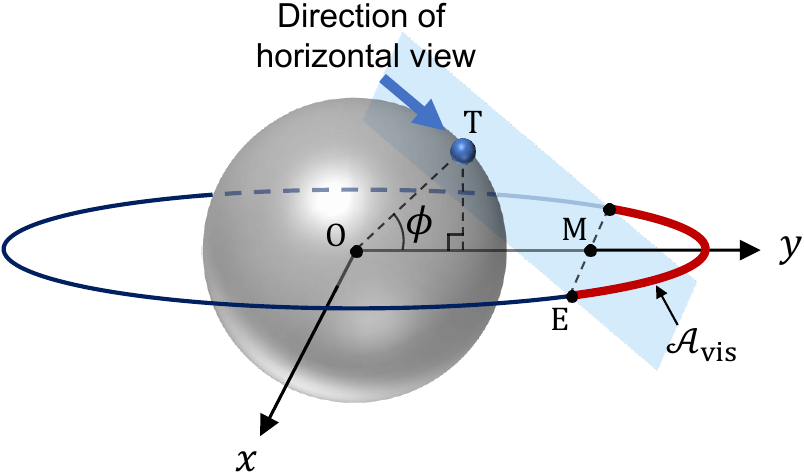} & 
     \includegraphics[width=0.600\columnwidth]{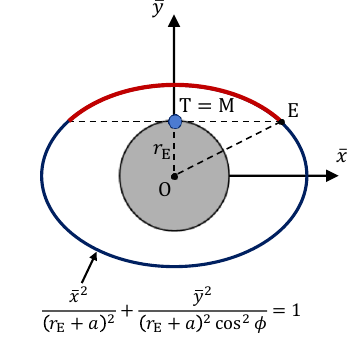} \\
     \midrule
     $\phi = \phiinv$ &
     \includegraphics[width=0.75\columnwidth]{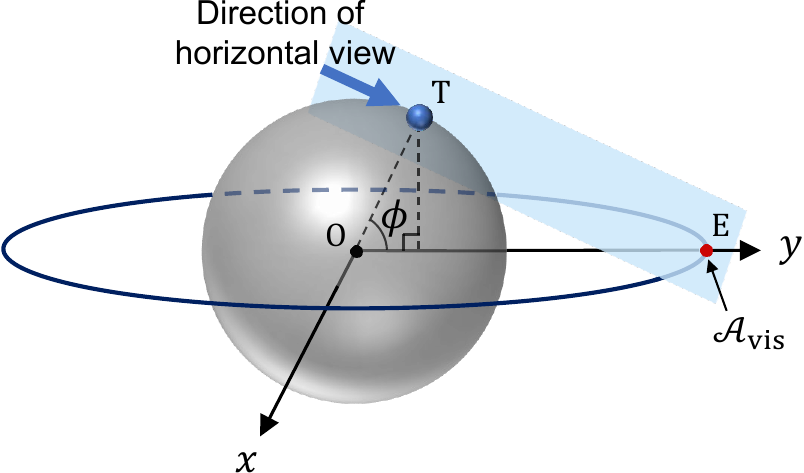} &      \includegraphics[width=0.600\columnwidth]{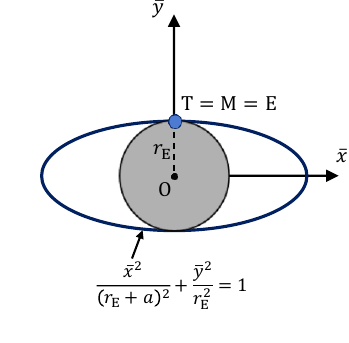}\\
     \bottomrule
     \end{tabular}
    \end{threeparttable}
     \caption{Latitude-dependent orbit visibility and horizontal view. The longitude of the terminal is set to $90$ degrees. The gray sphere and the dark blue circle on the $xy-$plane are Earth and the geostationary orbit, respectively, where the points $\mathrm{O}$ and $\mathrm{T}$ represent the Earth center and the terminal position, respectively. 
      The visible arc of the orbit is indicated as the red curve where the points $\mathrm{E}$ and $\mathrm{M}$ are an endpoint of the arc and the midpoint between two endpoints, respectively. The light blue rectangle represents the horizontal plane at the terminal.}\label{Fig:Orbit_visibility}
\end{figure*}

\vspace{0.7cm}

\section{Network Model}\label{sec:System_model}
We consider a downlink GEO satellite network where $N$ GEO satellites at altitude $\ag$ serve ground terminals.
It is notable that for the geostationary orbit, unlike LEOs, the altitude $\ag$ is consistently specified as 35,786 kilometers.
We assume that the positions of the GEO satellites $\mathbf{x}_n$, $n\in\{0,1,\!\cdots\!,N-1\}$, are randomly determined according to a homogeneous BPP $\BPPg =\{\mathbf{x}_0, \mathbf{x}_1, \cdots, \mathbf{x}_{N-1}\}$ in the circular geostationary orbit $\Ag$ as shown in Fig.~\ref{Fig:GEONet}.
This orbit can be expressed with spherical coordinates as $\Ag=\left\{\rho=\re+\ag, \psi=\frac{\pi}{2}, 0 \le\varphi\le 2\pi\right\}$, where  $\rho$, $\psi$, and $\varphi$ are the radial distance, polar angle, and azimuthal angle, respectively, and $\re$ is the Earth's radius.
We focus on a typical terminal located at arbitrary latitude $\phi$ and longitude~$\theta$ where its position is given by
\begin{align}
    \mathbf{t}
    =\begin{bmatrix}
    t_x \\
    t_y \\
    t_z
    \end{bmatrix}
    =\begin{bmatrix}
        \re \cos{\phi}\cos{\theta} \\ 
        \re \cos{\phi}\sin{\theta} \\ 
        \re \sin{\phi}
    \end{bmatrix}.
\end{align}

Because the satellites above the horizontal plane can be observed from Earth, a part of the geostationary orbit is only visible to the terminal, which we call the \textit{visible arc} and denote it by $\Agvis$.
We assume that the terminal is served by the nearest satellite in $\Agvis$, and the other satellites in $\Agvis$ become interfering nodes.
Let $\BPPg(\mathcal{X})$ denote the number of satellites distributed in region $\mathcal{X}$ according to the BPP $\BPPg$. Then $\BPPg(\Agvis)$ is the number of visible satellites positioned in the visible arc.
For notational simplicity, we let the index $n=0$ denote the serving satellite and the indices $n=1,\cdots,\BPPg(\Agvis)-1$ represent the interfering satellites, while the remaining indices are for the invisible satellites, which are irrelevant to the typical terminal.

The satellites adopt directional beamforming to compensate for large path losses at the receivers. 
For analytical tractability, we assume that the boresight of the serving satellite's beam is directed toward the target terminals, while the beams of other satellites are fairly misaligned [\ref{Ref:Okati1}], [\ref{Ref:Park}].
This assumption is well motivated in GEO scenarios because GEO satellites appear stationary from Earth, meaning that their beams remain fixed within a specific ground area.
With this property, the beams of GEO satellites are carefully designed to avoid overlapping the beams of the existing GEO satellites, thereby mitigating inter-satellite interference. 
Hence, the effective antenna gain $G_n$, $n\in\{0,1,\cdots,N-1\}$, is given by
\begin{align}
    G_n = 
    \begin{cases}
        \Gto\Gr, & \mbox{if  } n=0,\\
        \Gtn\Gr, & \mbox{otherwise }     
    \end{cases}
\end{align}
where $\Gtn$ is the transmit antenna gains of the satellites, and $\Gr$ is the receive antenna gain of the terminal.

The path loss between the terminal and the satellite at $\mathbf{x}_n$ is given by $\ell(\mathbf{x}_n)= \left(\frac{c} {4\pi\fc}\right)^2 R_n^{-\alpha}$ where $R_n=||\mathbf{x}_n-\mathbf{t}||$ is the distance, $c$ is the speed of light, $f_{\mathrm{c}}$ is the carrier frequency, and $\alpha$ is the path-loss exponent.
We assume the satellite channels experience Nakagami-$m$ fading, which effectively captures the LOS property of satellite channels [\ref{Ref:Park}], [\ref{Ref:Jung3}]. The Nakagami-$m$ fading model can reflect various channel circumstances by varying the Nakagami parameter $m$. For example, the Nakagami-$m$ distribution becomes the Rayleigh distribution when $m=1$, and the Rician-$K$ distribution when $m=\frac{(K+1)^2}{2K+1}$. The CDF of the channel gain of the Nakagami-$m$ fading model is given by $F_{h_n}(x)=1-e^{-m x}\sum_{k=0}^{m-1}\frac{(mx)^k}{k!}$.

Since the satellites are distributed according to the BPP, there can be three distribution cases for the visible satellites as follows.
\begin{itemize}
    \item Case 1: $\BPPg(\Agvis)=0$, i.e., no visible satellite exists. There are no serving and interfering satellites.
    \item Case 2: $\BPPg(\Agvis)=~1$, i.e.,  one visible satellite exists. The only visible satellite functions as the serving satellite without any interfering satellite.
    \item Case 3: $\BPPg(\Agvis)>1$, i.e., more than one visible satellite exist. Both the serving and at least one interfering satellites exist.
\end{itemize}
Considering these cases, the received signal-to-interference-plus-noise ratio ($\SINR$) at the typical terminal is given by
\begin{align}\label{eq:SINR}
    \SINR =
        \begin{cases} 
         0, & \mbox{if  } \BPPg(\Agvis) = 0 ,\\
        \frac{\Pt \Go h_0 \ell(\mathbf{x}_0)}{N_0 W}, & \mbox{if  } \BPPg(\Agvis) = 1,\\
        \frac{\Pt \Go h_0 \ell(\mathbf{x}_0)}{I + N_0 W}, & \mbox{if  } \BPPg(\Agvis) > 1
        \end{cases}
\end{align}
where $\Pt$ is the transmit power assuming all satellites transmit with the same power, and $I=\sum_{n=1}^{\BPPg(\Agvis)-1} \Pt \Gn h_n \ell(\mathbf{x}_n)$ is the aggregated inter-satellite interference.

\begin{figure}[!t]
\begin{center}
\includegraphics[width=\columnwidth]{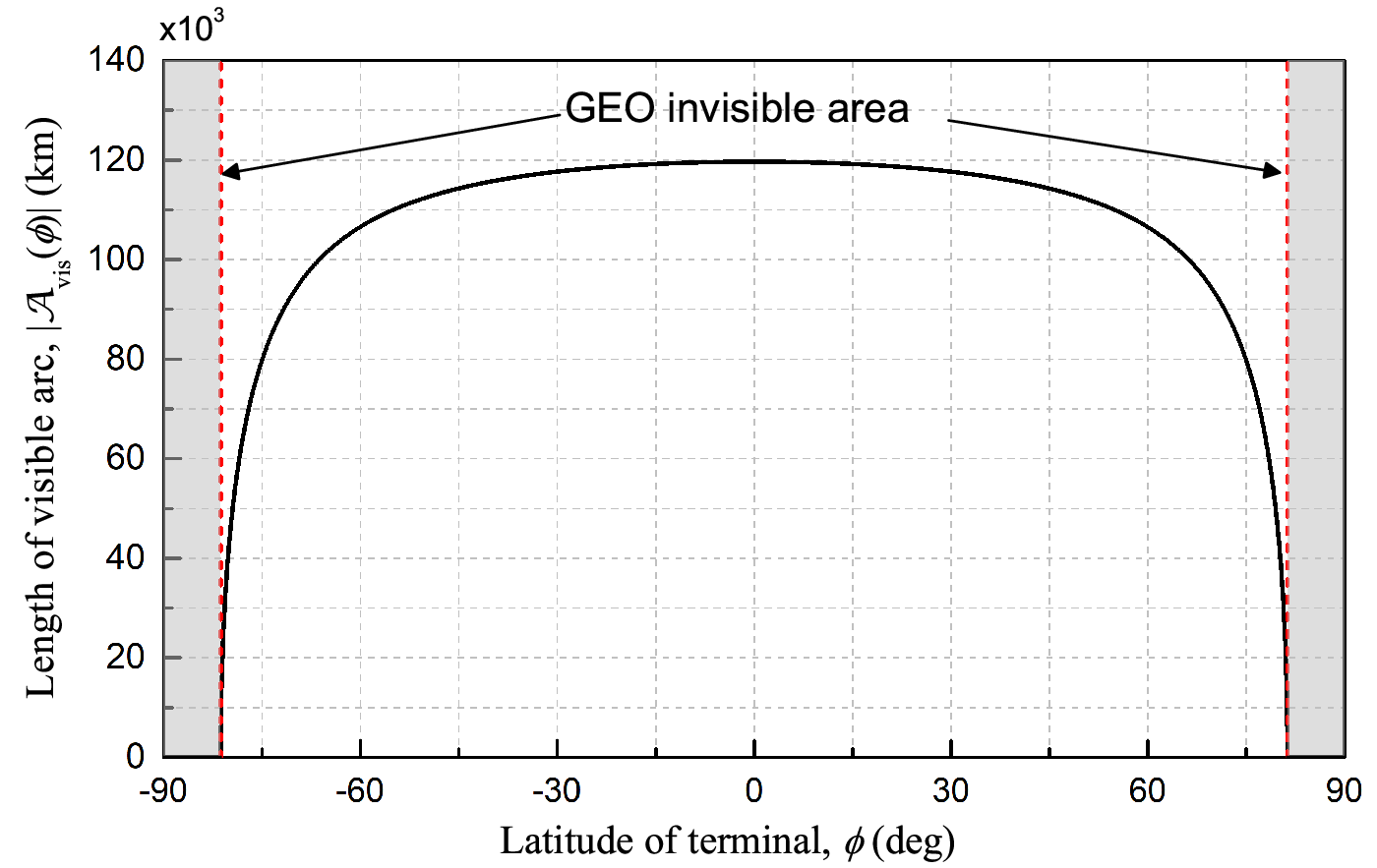}
\end{center}
\setlength\abovecaptionskip{.25ex plus .125ex minus .125ex}
\setlength\belowcaptionskip{.25ex plus .125ex minus .125ex}
\caption{Latitude-dependent length of the visible arc, $|\Agvis(\phi)|$.}
\label{Fig:Avis}
\end{figure}

\section{Mathematical Preliminaries} \label{sec:Math}
\subsection{Satellite Visibility Analyses}\label{sec:visibility_anal}
It is worth noting that the length of the visible arc $\Agvis(\phi)$ in the geostationary orbit highly depends on the terminal's latitude as shown in Fig. \ref{Fig:Orbit_visibility}. 
For example, when the terminal is placed on the equator, i.e., $\phi=0$, the visible arc, depicted as the red curve, is the longest. As the latitude increases, the visible arc shrinks and finally vanishes at the latitude of 
\begin{align}
    \phiinv=\cos^{-1}\left({\frac{\re}{\re+\ag}}\right) \approx 81.3 \text{ degrees.}\nonumber
\end{align} Based on this observation, the length of the visible arc is obtained in the following lemma.

\begin{lem}\label{lem:len_vis_arc}
    The length of the visible arc, i.e., $|\Agvis(\phi)|$, is given by
    \begin{align}
        |\Agvis(\phi)|= 2\left(\re+\ag\right)\sin^{-1}\left({\sqrt{1- \frac{\sec^2\phi}{(1+\ag/\re)^{2}}}}\right)
    \end{align}
    for $|\phi| < \phiinv$, and $|\Agvis(\phi)|=0$ otherwise.
\end{lem}

\begin{proof}[Proof:\nopunct]
As shown in Fig. \ref{Fig:Orbit_visibility}, if $\phi=0$, the geostationary orbit from the terminal's horizontal view becomes the circle with the radius of $\re+\ag$. However, if $\phi$ increases, the orbit can be seen as the ellipse whose semi-major and minor axes are $\re+\ag$ and $(\re+\ag)\cos\phi$, respectively. Hence, we can obtain the equation of the ellipse as
\begin{align}\label{eq:ellipse}
    \frac{\bar{x}^2}{(\re+\ag)^2}+\frac{\bar{y}^2}{(\re+\ag)^2\cos^2\phi}=1
\end{align}
where $\bar{x}$ and $\bar{y}$ are the projected axes observed from the horizontal view.
By substituting $\bar{y}=\re$ into \eqref{eq:ellipse}, we can obtain the length between $\mathrm{E}$ and $\mathrm{M}$ on the $xy-$plane as
\begin{align}\label{eq:EM}
    \overline{\mathrm{EM}}=\sqrt{(\re+\ag)^2-\re^2\sec^2\phi}.
\end{align}
Since $\overline{\mathrm{EO}}=\re+\ag$, we have
\begin{align}\label{eq:angleEOM}
    \angle{\mathrm{EOM}} = \sin^{-1}\left(\frac{\overline{\mathrm{EM}}}{\overline{\mathrm{EO}}}\right) = \sin^{-1}\left({\sqrt{1- \frac{\sec^2\phi}{(1+\ag/\re)^{2}}}}\right).
\end{align}
Using this angle, we finally obtain the length of the visible arc as $|\Agvis(\phi)|=2(\re+\ag)\angle{\mathrm{EOM}}$ for $\phi>0$. For $\phi<0$, the length can be readily obtained because the geometry is symmetric about the $xy-$plane, which completes the proof.
\end{proof}

\begin{rem}\label{rem:AvisInc}
    The length of the visible arc $\Agvis(\phi)$ is inversely proportional to the absolute value of the latitude $\phi$, i.e., $\Agvis(\phi) \propto \frac{1}{|\phi|}$, for $|\phi| \le \phiinv\approx 81.3$ degrees. This is because as $|\phi|$ increases from $0$ to $\phiinv$, $\sec\phi$ increases from $1$ to $1+\frac{\ag}{\re}$, resulting in a decrease in $\Agvis(\phi)$.    
    Hence, the maximum length of the visible arc is achieved when the terminal is located at the equator, i.e., $\phi=0$, and has a value of $2(\re+a)\cos^{-1}\left(\frac{\re}{\re+\ag}\right) \approx 119,657$ km. 
\end{rem}

\begin{rem}\label{rem:pvis}
    Based on the properties of the BPP, the number of satellites in the visible arc follows the binomial distribution with the success probability [\ref{Ref:Jung2}]
    \begin{align}\label{eq:Psucc_vis}
    \PsuccV=\frac{|\Agvis(\phi)|}{|\Ag|}=\frac{1}{\pi} \sin^{-1}\left({\sqrt{1- \frac{\sec^2\phi}{(1+\ag/\re)^{2}}}}\right)
    \end{align}
    for $|\phi| < \phiinv$, and $\PsuccV=0$ otherwise. 
    Thus, the average number of visible satellites is given by $\E[\BPPg(\Agvis)]=N \PsuccV$.
\end{rem}

Fig. \ref{Fig:Avis} shows the length of the visible arc depending on the latitude $\phi$. As expected in Remark~\ref{rem:AvisInc}, the length has its maximum $119,657$ km at $\phi=0$ degrees, decreases with $|\phi|$, and vanishes at  $|\phi| \approx 81.3$ degrees. This tendency implies that terminals near the equator are more likely to see many satellites, resulting in better satellite visibility, while those in the polar region, especially $|\phi|> 81.3$ degrees, are rarely in the coverage of GEO satellites.

\begin{figure*}[!t]
\centering
\includegraphics[width=2.1\columnwidth]{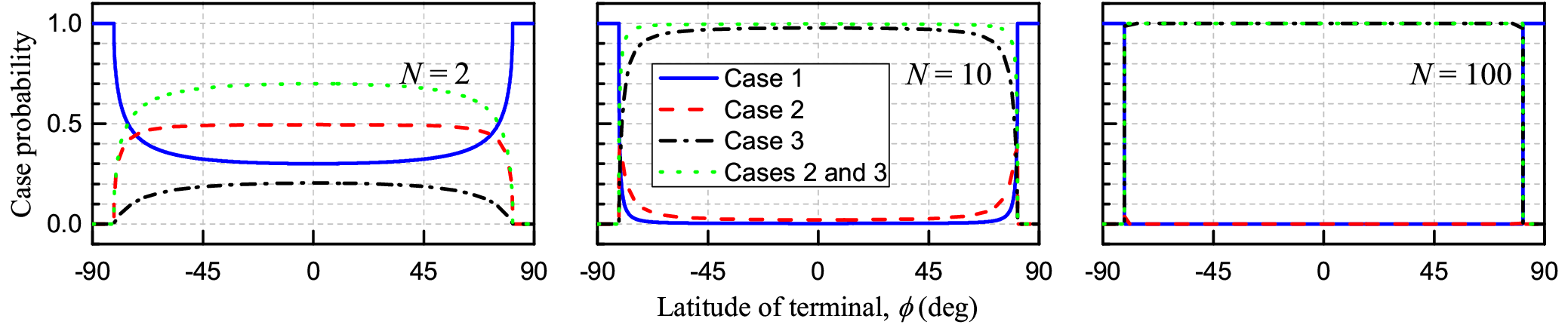}
\caption{Case probabilities versus the latitude of the terminal $\phi$ for $N=\{2,10,100\}$.}
\label{Fig:Pcases}
\end{figure*}

\begin{lem}\label{lem:Pcases}
    The probabilities for the visible satellite distribution of Cases 1, 2, and 3 are given by
    \begin{align}\label{eq:Pcase1}
        &\P[\BPPg(\Agvis)=0]= (1-\PsuccV)^{N},\nonumber\\
        &\P[\BPPg(\Agvis)=1]= N \PsuccV (1-\PsuccV)^{N-1},\nonumber\\
        &\P[\BPPg(\Agvis)>1]=1-(1-\PsuccV)^{N}-N \PsuccV (1-\PsuccV)^{N-1},
    \end{align}
    respectively.
\end{lem}
For more details, see Appendix \ref{App:Pcases}.
We remark that the probability of Case~1 is called the \textit{satellite-invisible probability}, which is the probability that all satellites are invisible, while the sum of the probabilities of Cases 2 and 3 is the \textit{satellite-visible probability}, which is the probability that at least one satellite is visible. 

\begin{rem}\label{rem:sat_inv_prob}
    For the northern hemisphere, i.e., $\phi > 0$, the satellite-invisible probability increases with the latitude $\phi$ and becomes one at $\phi = \phiinv$. This can be proved by taking the derivative of \eqref{eq:Pcase1} with respect to $\phi$ as
    \begin{align}\label{eq:der}
        \frac{d}{d\phi}(1-\PsuccV)^N 
        &= \frac{N(1-\PsuccV)^{N-1} \tan\phi}{\pi \sqrt{\frac{(1+\ag/\re)^2}{\sec^2\phi}- 1}}>0.
    \end{align}
    When $\phi>\phiinv$, $\PsuccV=0$, which results in the satellite-invisible probability equal to one.
\end{rem}

\begin{rem}\label{rem:Pcases_Ninf}
    When $N \rightarrow \infty$ for $|\phi|<\phiinv$, the terminal can see at least one serving satellite and one interfering satellite because $\P[\BPPg(\Agvis)=0] \rightarrow 0$, $\P[\BPPg(\Agvis)=1]  \rightarrow 0$, and $\P[\BPPg(\Agvis)>~1] \rightarrow 1$, which is intuitively true.
\end{rem}

Fig. \ref{Fig:Pcases} shows the case probabilities for various numbers of satellites $N=\{2,10,100\}$. This figure verifies Remark \ref{rem:sat_inv_prob} that the satellite-invisible probability corresponding to Case 1 increases with $\phi \in [0, \phiinv]$ and then reaches one at $\phi=\phiinv$.
As expected in Remark \ref{rem:Pcases_Ninf}, the satellite-visible probability corresponding to Cases 2 and 3 increases with $N$. Even with a hundred GEO satellites, $\P[\BPPg(\Agvis)>1]$ is almost one for any latitude less than $\phiinv$. Thus, Case 3 becomes the most probable case for terminals located at $|\phi|<\phiinv$ when a fairly large number of GEO satellites are in orbit. In contrast, only Case 1 happens for the terminals whose latitude is above $\phiinv$.

\subsection{Distance Distributions}
Next, we characterize the distributions of the following three distances from the typical terminal located at the latitude of~$\phi$.
\begin{itemize}
    \item $R\in[\rmin,\rmax]$: the distance to the nearest satellite
    \item $R_0\in[\rmin,\romax]$: the distance to the serving satellite
    \item $R_n\in[r_0,\romax]$: the distances to the interfering satellites given $R_0=r_0$, $n\in\left\{1,\cdots,\BPPg(\Agvis)-1\right\}$
\end{itemize}

The distances to the nearest and farthest points in the geostationary orbit are denoted by $\rmin$ and $\rmax$ and can be obtained by applying the Pythagorean theorem to $\triangle \mathrm{S_NTT}^\prime$ and $\triangle \mathrm{S_FTT}^\prime$ shown in Fig. \ref{Fig:GEO_geometry}, i.e.,
\begin{align}\label{eq:rmin_pt}
     \overline{\mathrm{S_N T}}^2 = \overline{\mathrm{S_N T^{\prime}}}^2 + \overline{\mathrm{T T^{\prime}}}^2  = \left(\re+\ag - \sqrt{t_x^2 + t_y^2}\right)^2 + t_z^2
\end{align}
and
\begin{align}\label{eq:rmax_pt}
    \overline{\mathrm{S_F T}}^2 = \overline{\mathrm{S_F T^{\prime}}}^2 + \overline{\mathrm{T T^{\prime}}}^2  = \left(\re+\ag + \sqrt{t_x^2 + t_y^2}\right)^2 + t_z^2.
\end{align}
After some manipulation using the facts that $t_x^2 + t_y^2=\re^2\cos^2\phi$ and $t_z^2=\re^2\sin^2\phi$, we can 
 obtain $\rmin$ and $\rmax$ as 
 \begin{align}
    \rmin=\overline{\mathrm{S_N T}}=\sqrt{(\re+\ag-\re\cos\phi)^2 + \re^2 \sin^2\phi}
\end{align}
and
\begin{align}
    \rmax=\overline{\mathrm{S_F T}}=\sqrt{(\re+\ag+\re\cos\phi)^2 + \re^2 \sin^2\phi}.
\end{align}
Moreover,
the maximum distance to the visible satellite, denoted by $\romax$, can be obtained as
$\romax = \overline{\mathrm{TE}} = \left(\overline{\mathrm{TM}}^2 + \overline{\mathrm{EM}}^2\right)^{1/2}$ according to the geometry in Fig.  \ref{Fig:GEO_geometry}. From $\overline{\mathrm{EM}}$ given in \eqref{eq:EM} with the fact that $\overline{\mathrm{TM}}=\re \tan\phi$, we have 
\begin{align}
    \romax=\sqrt{\re^2\tan^2\phi + (\re+\ag)^2 - \re^2 \sec^2\phi}=\sqrt{\ag^2+2\ag \re}.
\end{align}
\begin{rem}
The largest possible distance from the terminal to any visible satellite is always $\romax=\sqrt{\ag^2+2\ag \re}\approx 41,679$ km regardless of $\phi$. This result is intuitive because the endpoints of the visible arc can be seen as the intersection between the horizontal plane and the sphere with the radius of $\re+\ag$.
\end{rem}

\begin{figure}[!t]
\begin{center}
\includegraphics[width=.97\columnwidth]{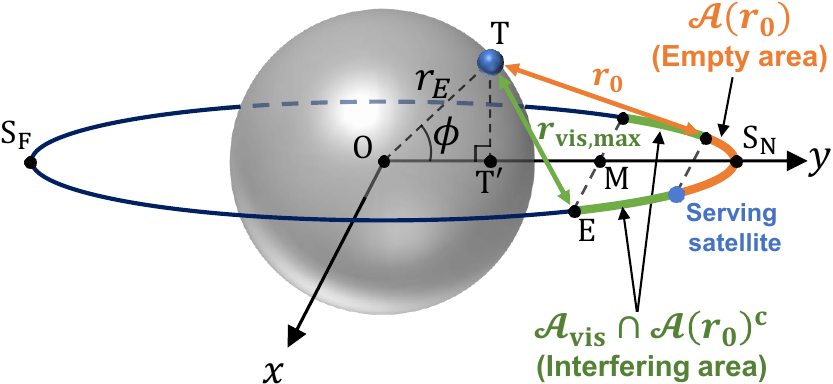}
\end{center}
\setlength\abovecaptionskip{.25ex plus .125ex minus .125ex}
\setlength\belowcaptionskip{.25ex plus .125ex minus .125ex}
\caption{The relationship between distances and areas. The points $\mathrm{S_N}$ and $\mathrm{S_F}$ are the nearest and farthest points in the orbit, and $\mathrm{T}^{\prime}$ is the projection of the terminal's position $\mathrm{T}$ onto the $xy-$plane.}
\label{Fig:GEO_geometry}
\end{figure}

Before deriving the distance distributions, we let $\Ag(r)$, $r\in[\rmin,\rmax]$, denote an arc of the geostationary orbit whose maximum distance to the terminal is $r$, which is shown in Fig. \ref{Fig:GEO_geometry} with $r=r_0$.
Using the law of cosines, the length of the arc $\Ag(r)$ is given by
\begin{align}\label{eq:lenBr}
    |\Ag(r)|
    &=2(\re+\ag)\angle{\mathrm{EOT}^{\prime}}\nonumber\\
    &=2(\re+\ag)\cos^{-1}\left(\frac{\overline{\mathrm{OE}}^2 + \overline{\mathrm{OT}^{\prime}}^2 - \overline{\mathrm{ET}^{\prime}}^2}{2 \overline{\mathrm{OE}} \cdot \overline{\mathrm{OT}^{\prime}}}\right)\nonumber\\
    &=2\pi(\re+\ag)\cdot\underbrace{\frac{1}{\pi}\cos^{-1}\left(\frac{(\re+\ag)^2+\re^2-r^2}{2(\re+\ag)\re\cos\phi}\right)}_{\triangleq\Psi(r,\phi)}
\end{align} 
where $\overline{\mathrm{OE}}=\re+\ag$, $\overline{\mathrm{OT}^{\prime}}=\re \cos\phi$, and $\overline{\mathrm{ET}^{\prime}}=\sqrt{r^2-t_z^2}=\sqrt{r^2-\re^2\sin\phi}$.
Please note that we define a new function $\Psi(r,\phi)$ for the simplicity of notation, which will be used to efficiently express the distance distributions in the following lemmas.

\begin{figure*}[!t]
\begin{center}
\includegraphics[width=2\columnwidth]{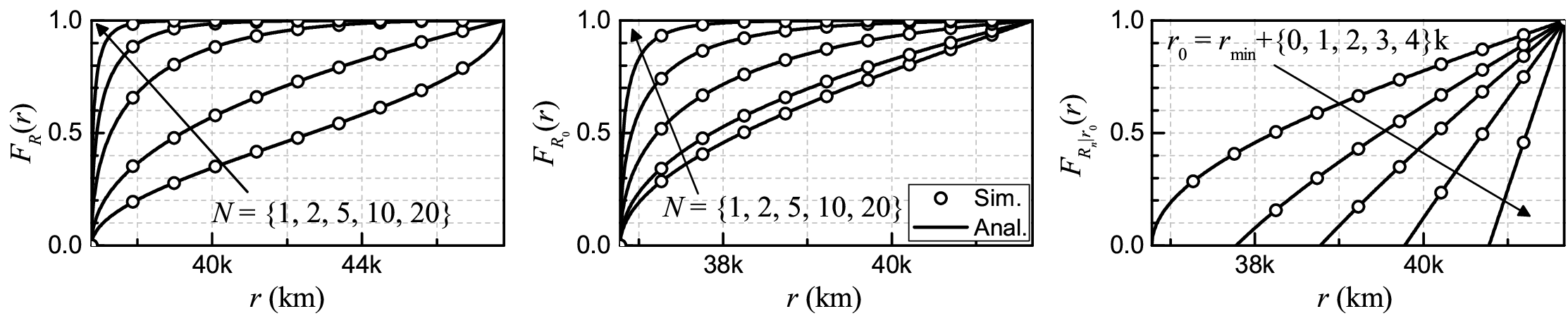}
\end{center}
\setlength\abovecaptionskip{.25ex plus .125ex minus .125ex}
\setlength\belowcaptionskip{.25ex plus .125ex minus .125ex}
\caption{The analytical and numerically derived CDFs of $R$, $R_0$, and $R_n$ for $\phi=30$ degrees.}
\label{Fig:CDFs}
\end{figure*}

\begin{lem}\label{lem:CDFR}
    The CDF and PDF of $R$ are respectively given by
    \begin{align}\label{eq:CDFR}
        F_R(r) = 
        \begin{cases} 
         0, & \mbox{if  } r<\rmin,\\
        1-\left(1-\Psi(r,\phi)\right)^N, & \mbox{if  } \rmin \le r < \rmax,\\
        1, & \mbox{otherwise}
        \end{cases}
    \end{align}
    and
    \begin{align}\label{eq:PDFR}
    f_R(r) = 
        \begin{cases} 
         \frac{2 N r \left(1-\PsiG(r,\phi)\right)^{N-1}}{ \pi\sqrt{v_1-\left(v_2-r^2\right)^2}}, & \mbox{if  } \rmin \le r < \rmax,\\
        0, & \mbox{otherwise}
        \end{cases}
    \end{align}
    where 
    $v_1=4(\re+\ag)^2\re^2\cos^2\phi$ and $v_2=(\re+\ag)^2+\re^2$.
\end{lem}

For more details, see Appendix \ref{App:CDFR}.

\begin{lem}\label{lem:CDFR0}
    The CDF and PDF of $R_0$ are respectively given by
    \begin{align}\label{eq:CDFR0}
        F_{R_0}(r) =
            \begin{cases} 
            0, & \mbox{if  } r<\rmin,\\
             \frac{F_R(r)}{F_R( \romax)}, & \mbox{if  } \rmin \le r < \romax,\\
            1, & \mbox{otherwise}
            \end{cases}
    \end{align}
    and 
    \begin{align}\label{eq:PDFR0}
        f_{R_0}(r) =
            \begin{cases} 
             \frac{f_R(r)}{F_R( \romax)}, & \mbox{if  } \rmin \le r < \romax,\\
            0, & \mbox{otherwise.}
            \end{cases}
    \end{align}
\end{lem}

For more details, see Appendix \ref{App:CDFR0}.

 \begin{rem}
     When $N \rightarrow \infty$, the CDF and PDF of the distance $R_0$ asymptotically become $F_{R_0}(r) \rightarrow u(r-\rmin)$ and $f_{R_0}(r) \rightarrow \delta(r-\rmin)$, respectively. This means that the distance to the serving satellite is deterministic and has a value of $\rmin$, i.e., the possible minimum distance to satellites given the latitude $\phi$.
 \end{rem}

\begin{lem}\label{lem:CDFRn}
    Given $R_0=r_0$, the CDF and PDF of $R_n$, $n\in\{1,2,\cdots,\BPPg(\Agvis)-1\}$, are, respectively, given by
    \begin{align}\label{eq:CDFRn}
        F_{R_n|r_0}(r) = 
        \begin{cases} 
            0, & \mbox{if  } r<r_0,\\
             \frac{\Psi(r,\phi)-\Psi(r_0,\phi)}{\Psi(\romax,\phi)-\Psi(r_0,\phi)}, & \mbox{if  } r_0 \le r < \romax,\\
            1, & \mbox{otherwise}
        \end{cases}
    \end{align}
    and 
    \begin{align}\label{eq:PDFRn}
    f_{R_n|r_0}(r) = 
    \begin{cases} 
         \frac{2r/\left(\pi \sqrt{v_1-\left(v_2-r^2\right)^2}\right)}{\Psi(\romax,\phi)-\Psi(r_0,\phi)}, & \mbox{if  } r_0 \le r < \romax,\\
        0, & \mbox{otherwise}.
    \end{cases}
\end{align}
\end{lem}

For more details, see Appendix \ref{App:CDFRn}.
\\

Fig. \ref{Fig:CDFs} verifies the analytical expressions given in Lemmas \ref{lem:CDFR}, \ref{lem:CDFR0}, and \ref{lem:CDFRn} for $\phi=30$ degrees. It is shown that our analyses are in good agreement with the simulation results. 
The case probabilities and the distance distributions obtained in this section are the key instruments, which will be used for deriving the stochastic geometry-based performance in the next section.

\section{Coverage Analysis} \label{sec:Coverage}
In this section, we first derive the coverage probability of the GEO satellite network with the BPP-based satellite distribution. Then, we further simplify the expression using the Poisson limit theorem, which states that the binomial distribution can be approximated as the Poisson distribution as $N\rightarrow \infty$ [\ref{Ref:Book:Chiu}].

\subsection{Binomial Distribution-Based Analysis}
Before analyzing the coverage probability, we derive the Laplace transform of the aggregated interference power using the following two lemmas.

\begin{lem}\label{lem:num_int_sat}
        Given that the serving satellite is at a distance of $r_0$ from the terminal, the number of possible interfering satellites is $\BPPg(\Ag\cap\Ag(r_0)^{\mathrm{c}}) =N-1$ (except for the serving satellite), and the region where interfering satellites can be located is $\Agvis \cap \Ag(r_0)^{\mathrm{c}}$. The number of interfering satellites in this region is a  binomial random variable with the success probability $\PsuccI$, i.e., $\BPPg(\Agvis \cap \Ag(r_0)^{\mathrm{c}}) \sim \mathrm{Bin}(N-1, \PsuccI)$,
        where 
    \begin{align}\label{eq:PsuccI}
    \PsuccI = \frac{|\Agvis\cap\Ag(r_0)^{\mathrm{c}}|}{|\Ag\cap\Ag(r_0)^{\mathrm{c}}|} = \frac{\Psi(\romax,\phi)-\Psi(r_0,\phi)}{1-\Psi(r_0,\phi)}.
    \end{align}
\end{lem}

This result comes directly from the definition of the BPP. For more details, see [\ref{Ref:Book:Chiu}].

\begin{lem}\label{lem:LI}
Given $R_0=r_0$, the Laplace transform of the aggregated interference power $I=\sum_{n=1}^{\BPPg(\Agvis\cap\Ag(r_0)^{\mathrm{c}})} \Pt \Gn h_n \ell(\mathbf{x}_n)$ is given by
\begin{align}
    \mathcal{L}_{I|r_0}(s)
    &=\sum_{\nI=0}^{N-1} \binom{N-1}{\nI}\PsuccI^{\nI}\PsuccI^{N-1-\nI} \nonumber\\
    &\times   \prod_{n=1}^{\nI}\int_{r_0}^{\romax}\left(\frac{m \un r_n^{\alpha}}{s+m \un r_n^{\alpha}}\right)^m f_{R_n|r_0}(r_n) dr_n    
\end{align}
    where $\un = 16 \pi^2\fc^2/(\Pt \Gn c^2)$.
\end{lem}

For more details, see Appendix \ref{App:LTG}.\\


The coverage probability is the probability that the $\SINR$ at the typical terminal defined in \eqref{eq:SINR} is greater than or equal to a threshold $\tau$, i.e., $\P[\SINR\geq \tau]$. 
Using the result in Lemma~\ref{lem:LI}, the coverage probability is given in the following theorem.

\begin{thm}\label{thm:Pcov}
    The coverage probability for GEO satellite networks is approximated as
    \begin{align}\label{eq:Pcov}
    \Pcov&(\tau;m)
        \approx (1-(1-\PsuccV)^{N}) \sum_{i=1}^{m}\binom{m}{i}(-1)^{i+1} \nonumber\\
        &\times\int_{\rmin}^{\romax} e^{-\nu i \uo N_0 W \tau r^{\alpha}}\mathcal{L}_{I|r_0=r}(\nu i \uo \tau r^{\alpha})f_{R_0}(r)dr
    \end{align}
    where $\nu=m(m!)^{-1/m}$.
\end{thm}

For more details, see Appendix \ref{App:Pcov}.\\

This approximated coverage probability in Theorem~\ref{thm:Pcov} becomes exact when $m=1$, which is given in the following corollary.

\begin{cor}
    Under the Rayleigh fading, i.e., $m=1$, the coverage probability for GEO satellite networks in Theorem~\ref{thm:Pcov} becomes exact and is given by
    \begin{align}\label{eq:Pcov_m=1}
    \Pcov(&\tau;1)
        = (1-(1-\PsuccV)^{N}) \nonumber\\
        &\times\int_{\rmin}^{\romax} e^{-\uo N_0 W \tau r^{\alpha}}\mathcal{L}_{I|r_0=r}(\uo \tau r^{\alpha})f_{R_0}(r)dr.
    \end{align}
\end{cor}
This result is obtained by directly setting $m=1$ in \eqref{eq:Pcov}.
The expression for the coverage probability in \eqref{eq:Pcov} includes the integrals that appear to be unsolvable analytically, but can be evaluated numerically. To further simplify the expression, we conduct the Poisson limit theorem-based approximation in the following section. 

\subsection{Poisson Limit Theorem-Based Approximation}\label{sec:PLT}
When sufficiently many GEO satellites are in orbit, e.g., $N\rightarrow\infty$, the BPP can be interpreted as a PPP $\bar{\BPPg}$ with the density of 
\begin{align}
    \lambda=\frac{N}{|\Ag|}=\frac{N}{2\pi(\re+\ag)}
\end{align} 
according to the Poisson limit theorem [\ref{Ref:Book:Chiu}]. 
The void probability of the PPP $\bar{\BPPg}$ for arc $\Ag(r)$ is the probability that there is no satellite in $\Ag(r)$, which is given by $e^{-\lambda |\Ag(r)|}$. 
Using this property, we obtain the satellite visible probability as
\begin{align}\label{eq:Pvis_approx}
    \P[\BPPg&(\Agvis)>0] 
    = 1-\P[\BPPg(\Agvis)=0]\nonumber\\
    &\approx 1-e^{-\lambda |\Ag(\romax)|} = 1-e^{-N \Psi(\romax,\phi)}
\end{align}
and approximate the CDF and PDF of $R$ and $R_0$ in the next two lemmas.

\begin{lem}\label{lem:CDFR_approx}
    The approximated CDF and PDF of $R$, denoted by $\tilde{F}_{R}(r)$ and $\tilde{f}_{R}(r)$, are respectively given by
    \begin{align}\label{eq:CDFR_approx}
    \tilde{F}_R(r) =
    \begin{cases} 
         0, & \mbox{if  } r<\rmin,\\
        1-e^{-N \Psi(r,\phi)}, & \mbox{if  } \rmin \le r < \rmax,\\
        1, & \mbox{otherwise}
    \end{cases}
    \end{align}
    and 
    \begin{align}\label{eq:PDFR_approx}
    \tilde{f}_R(r) = 
        \begin{cases} 
          \frac{2 r N e^{-N \Psi(r,\phi)}}{\pi\sqrt{v_1-\left(v_2-r^2\right)^2}}, & \mbox{if  } \rmin \le r < \rmax,\\
        0, & \mbox{otherwise}.
        \end{cases}
    \end{align}
\end{lem}

\begin{proof}[Proof:\nopunct]
    With the PPP approximation of GEO satellite positions, we approximate the CDF of $R$ as
    \begin{align}\label{eq:CDFR_approx_1}
        \tilde{F}_R(r) = \P[R\le r] = 1-\P[R>r] = 1-e^{-\lambda |\Ag(r)|}.
    \end{align}
    By substituting \eqref{eq:lenBr} and $\lambda=\frac{N}{2\pi(\re+\ag)}$ here, the approximated CDF is obtained. 
    The PDF is directly given by taking the derivative of the CDF, which completes the proof.
\end{proof}

\begin{lem}\label{lem:CDFR0_approx}
    The approximated CDF and PDF of $R_0$, denoted by $\tilde{F}_{R_0}(r)$ and $\tilde{f}_{R_0}(r)$, are obtained by substituting the CDF and PDF of $R$ in Lemma \ref{lem:CDFR_approx} 
    into \eqref{eq:CDFR0} and \eqref{eq:PDFR0} as
    \begin{align}\label{eq:CDFR0_approx}
    \tilde{F}_{R_0}(r) =
        \begin{cases} 
            0, & \mbox{if  } r<\rmin,\\
            \frac{ 1-e^{-N \Psi(r,\phi)}}{ 1-e^{-N \Psi(\romax,\phi)}}, & \mbox{if  } \rmin \le r < \romax,\\
            1, & \mbox{otherwise}
        \end{cases}
    \end{align}
    and 
    \begin{align}\label{eq:PDFR0_approx}
    &\tilde{f}_{R_0}(r) \nonumber\\
    &=
        \begin{cases} 
        \frac{2 N / \pi}{ 1-e^{-N \Psi(\romax,\phi)}} \cdot \frac{ r  e^{-N \Psi(r,\phi)}}{\sqrt{v_1-\left(v_2-r^2\right)^2}}, & \mbox{if  } \rmin \le r < \romax,\\
        0, & \mbox{otherwise}.
        \end{cases}
\end{align}
\end{lem}
The proof of this lemma is complete by following the same steps as in the proof of Lemma \ref{lem:CDFR0}.
These distance distributions are much simpler than the BPP-based results obtained in the previous section. Thus, the simplified CDFs and PDFs will play a crucial role in reducing the computational complexity of evaluating the coverage probability. Also, the simplified Laplace transform is provided next. 
\begin{lem}\label{lem:LI_approx}
    When the satellites are distributed according to the PPP $\bar{\BPPg}$ with the density of $\lambda=\frac{N}{2\pi(\re+\ag)}$, and the effective antenna gains for the interfering satellites are equal, i.e., $\Gn=\bar{G}\ \forall n \neq 0$ for an arbitrary constant $\bar{G}$, the Laplace transform of the aggregated interference power is derived as
    \begin{align}\label{eq:LI_approx}
    &\tilde{\mathcal{L}}_{I|r_0}(s)
        =e^{-\frac{2N}{\pi} \left(\Omega_1(\romax)-\Omega_1(r_0) - \Omega_2(s,r_0) \right) }
    \end{align}
    where $\Omega_1(r)=-\frac{1}{2}\tan^{-1}\left(\frac{v_2 - r^2}{\sqrt{v_1-(v_2-r^2)^2}}\right)$ and $\Omega_2(s,r_0)=\int_{r_0}^{\romax} \frac{1}{\left(\frac{s r^{-\alpha}}{m \omega}+1\right)^{m}} \frac{ r dr}{\sqrt{v_1-\left(v_2-r^2\right)^2}}$ with $\omega = \frac{16 \pi^2\fc^2}{\Pt \bar{G} c^2}$.
\end{lem}

\begin{proof}[Proof:\nopunct]
   See Appendix \ref{App:LI_approx_proof}.
\end{proof}

\begin{rem}
    Unlike the Laplace transform of the aggregated interference power for the binomially distributed satellites, given in Lemma \ref{lem:LI}, the approximated Laplace transform in Lemma \ref{lem:LI_approx} does not rely on any distance distribution thanks to the stochastic property of the PPP $\bar{\BPPg}$. 
\end{rem}

Using the simplified distance distributions and Laplace transform, the coverage probability is obtained in the following theorem.
\begin{thm}\label{thm:Pcov_PPP}
    When GEO satellites are distributed according to a PPP $\bar{\BPPg}$ with a density of $\lambda=\frac{N}{2\pi(\re+\ag)}$, the coverage probability is given by
    \begin{align}\label{eq:Pcov_approx}
    \tilde{P}_{\mathrm{cov}}(\tau;m) \approx \frac{2N/\pi}{1-e^{-N \Psi(\romax,\phi)}}  \sum_{i=1}^{m} \binom{m}{i} (-1)^{i+1}\Xi_i(\tau;m)
    \end{align}
    where 
    \begin{align}
        \Xi_i(\tau;m) = \int_{\rmin}^{\romax}\!\frac{r e^{-\Theta_i(r,\tau;m)} }{\sqrt{v_1-\left(v_2-r^2\right)^2}} dr
    \end{align}
    with
    $\Theta_i(r,\tau;m) 
        = N \Psi(r,\phi) + \nu i \uo N_0 W \tau r^{\alpha}+ \frac{2N}{\pi} (\Omega_1(\romax)-\Omega_1(r_0) - \Omega_2(\nu i \uo \tau r^{\alpha},r_0) ).$
\end{thm}

\begin{proof}[Proof:\nopunct]
The proof is complete by following the proof of Theorem \ref{thm:Pcov} with the approximated results  \eqref{eq:Pvis_approx}, \eqref{eq:PDFR0_approx}, and \eqref{eq:LI_approx}.
\end{proof}

Although the expression in Theorem \ref{thm:Pcov_PPP} has the integral term in $\Xi_i(\tau;m)$, it is much easier to calculate than that in Theorem \ref{thm:Pcov} due to the simplified Laplace transform.

\section{Simulation Results}\label{sec:Sim_results}

\begin{table}
    \caption{Simulation Parameters}\label{Table:Sim_Param}
    \centering
    \begin{tabular}{|l|c|}
     \hline 
     Parameter & Value\\
     \hline\hline
     Radius of Earth $r_{\mathrm{e}}$ & 6,378 km\\
     \hline
     Speed of light $c$  & $3 \times 10^8$ m/s\\
     \hline
     Noise spectral density $N_0$ & $-174$ dBm/Hz\\
     \hline

     Altitude $a$ & 35,786 km\\
     \hline
     Carrier frequency $f_{\mathrm{c}}$ & 2 GHz \\
     \hline
     
     Path-loss exponent $\alpha$ & 3\\
     \hline
     Transmit antenna gain of the serving satellite $\Gto$ & 51 dBi\\
     \hline
     Receive antenna gain of the terminal $\Gr$ & 0 dBi\\
     \hline
     EIRP density & 59 dBW/MHz\\
     \hline
     Bandwidth $W$ & 30 MHz\\
     \hline
    \end{tabular}
\end{table}

In this section, we numerically verify the derived results based on the simulation parameters listed in Table \ref{Table:Sim_Param} unless otherwise stated.
The handheld terminals are considered for the S-band as in the 3GPP standardization [\ref{Ref:3GPP_38.821}]. 
With the assumed effective isotropically radiated power (EIRP) density, which is calculated as $\Pt \Gto/W$, we obtain the transmit power of the satellites as $\Pt =52.77$ dBm. 
The typical terminal is located in Seoul, South Korea, i.e., $\{\phi,\theta\}=\{37,137\}$ deg.


Fig. \ref{Fig:comparison} compares our analysis with the simulation results considering the actual GEO satellites. The positions of the actual GEO satellites, depicted in Fig.~\ref{Fig:comparison}~\subref{Fig:geo_scatter}, are calculated from the two-line element dataset given in https://celestrak.org/ on October 21, 2023. In this dataset, we consider the satellites with an inclination less than 1 degree among all the actual GEO satellites in geosynchronous orbits. As a result, the number of considered satellites is $N=391$.
In Fig.~\ref{Fig:comparison}~\subref{Fig:avg_num_vs_phi}, we compare the average number of visible satellites for the actual and BPP-based satellite distributions. For the actual distribution, the average number is calculated by averaging the number of visible satellites over all longitudes at each latitude.
For the BPP-based distribution, the  average number comes from $\E[\BPPg(\Agvis)]=N \PsuccV$ given in Remark~\ref{rem:pvis}.
It is shown that our approach to model GEO satellites is fairly reasonable because the average number of actually visible satellites at a terminal is almost the same as our analysis.
This alignment is mainly achieved by displacing adjacent GEO satellites at a certain distance to be almost evenly distributed in the geostationary orbit in order to minimize interference with other satellites.

\begin{figure*}
\centering
\subfigure[The distribution of actual GEO satellites on $xy-$plane\protect\footnotemark]{
\includegraphics[width=0.93\columnwidth]{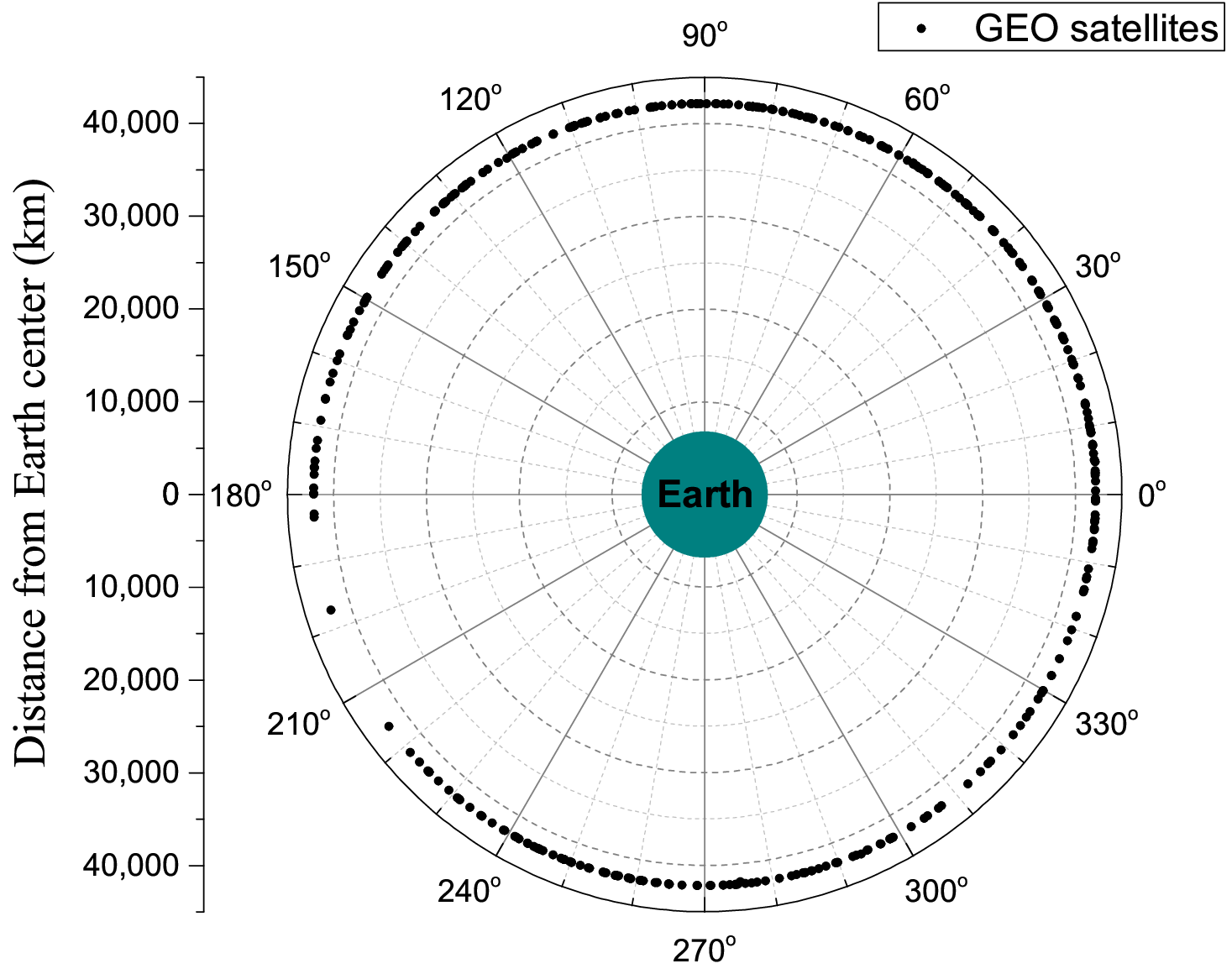}
\label{Fig:geo_scatter}
}
\subfigure[The average number of visible satellites when $N=391$]{
\includegraphics[width=0.98\columnwidth]{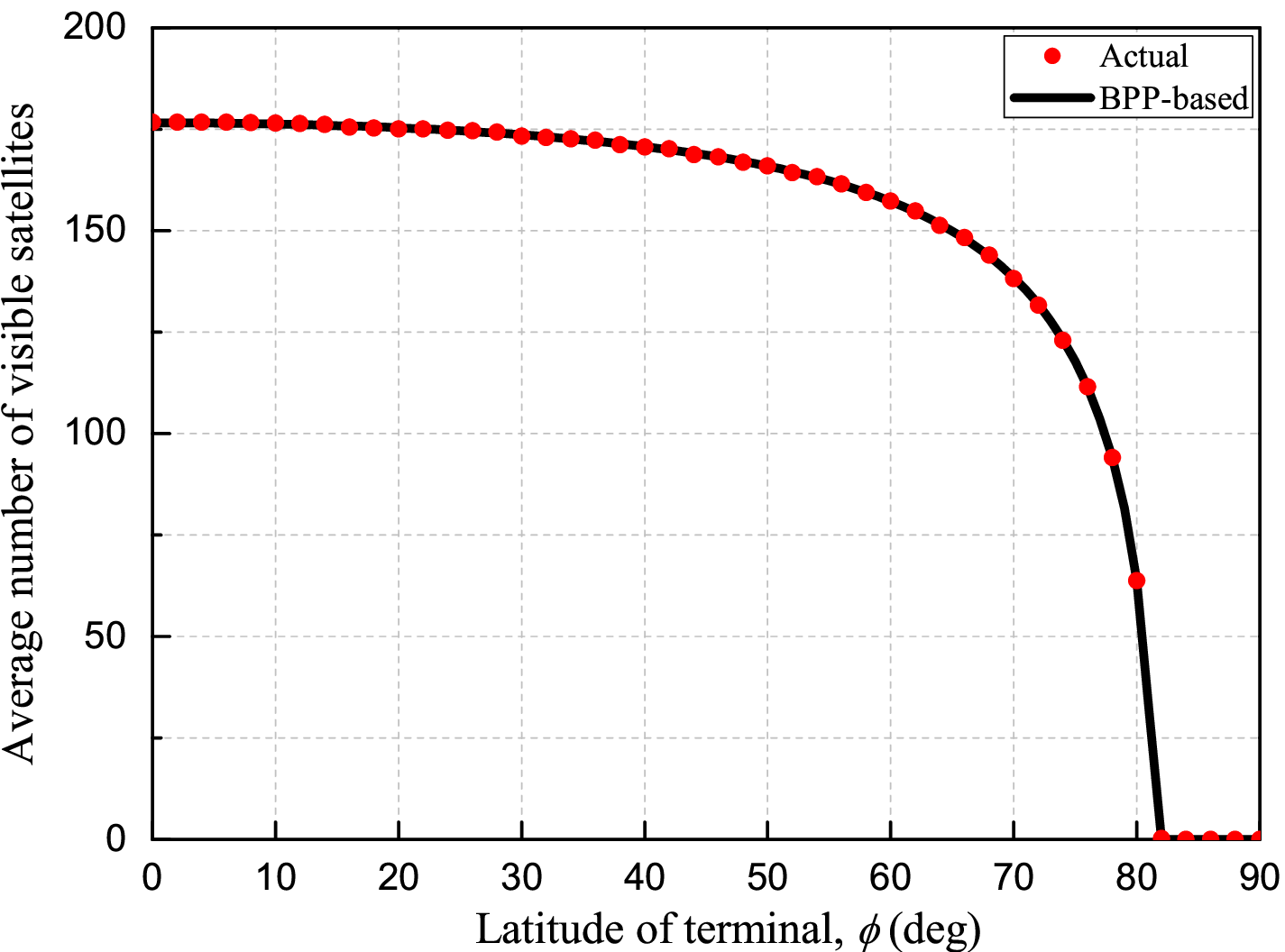}
\label{Fig:avg_num_vs_phi}
}
\caption{Comparison between the actual and the BPP-based GEO satellite distributions.}
\label{Fig:comparison}
\end{figure*}
\footnotetext{In Fig.~\ref{Fig:comparison}~\subref{Fig:geo_scatter}, a small number of GEO satellites are positioned in the areas with the longitude from 180 to 220 degrees ($140^{\circ}$W-$180^{\circ}$W). These areas encompass the Alaska and Pacific region where the demand for communication services is scarce. In the future, additional GEO satellites may be deployed in these areas to accommodate potential service needs.}

\begin{figure*}
\centering
\subfigure[$N=100$, $\Go/\Gn=20$ dB]{
\includegraphics[width=.97\columnwidth]{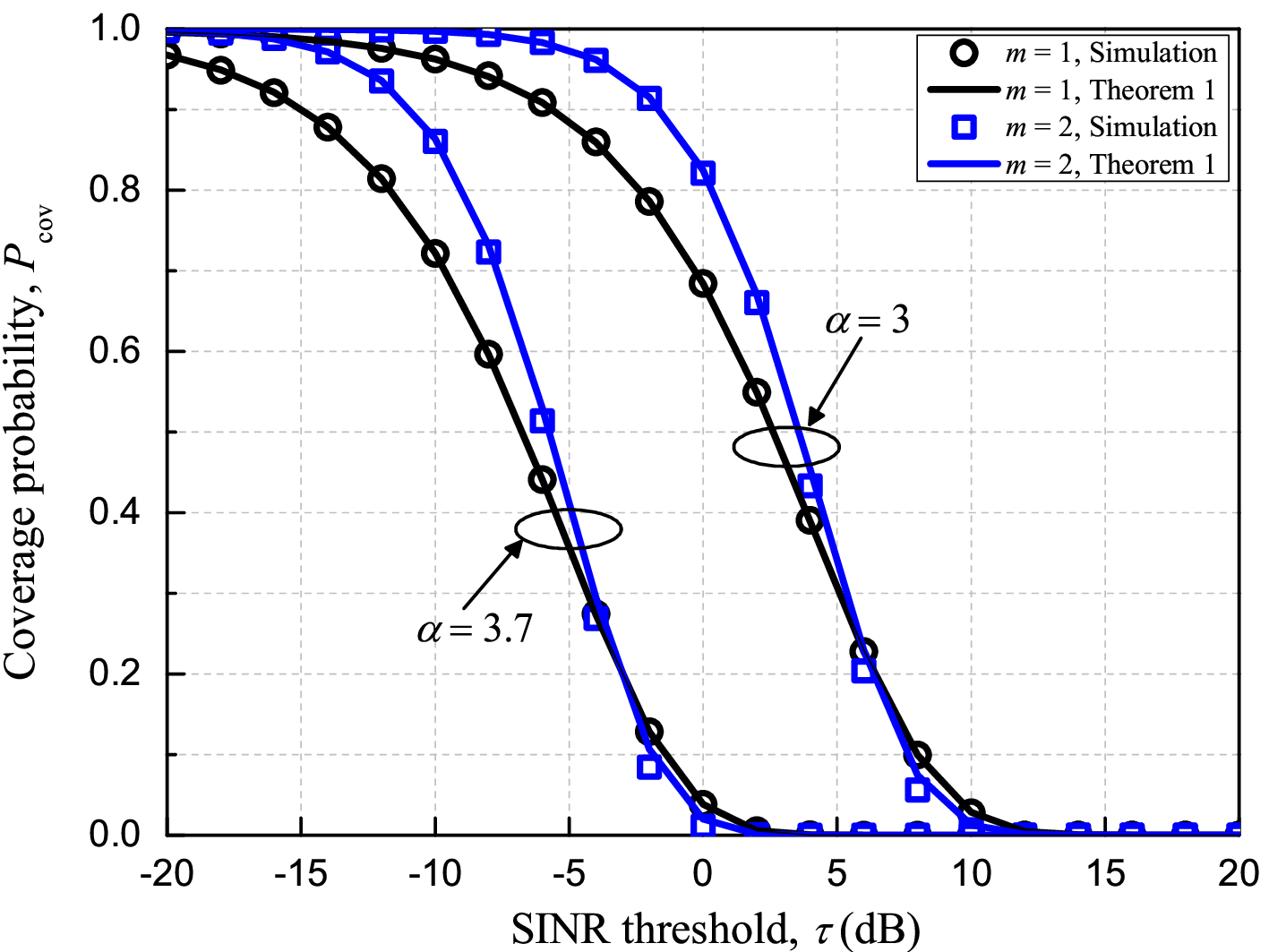}
\label{Fig:Pcov_vs_tau}
}
\subfigure[$m=2$, $\Go/\Gn=30$ dB]{
\includegraphics[width=.97\columnwidth]{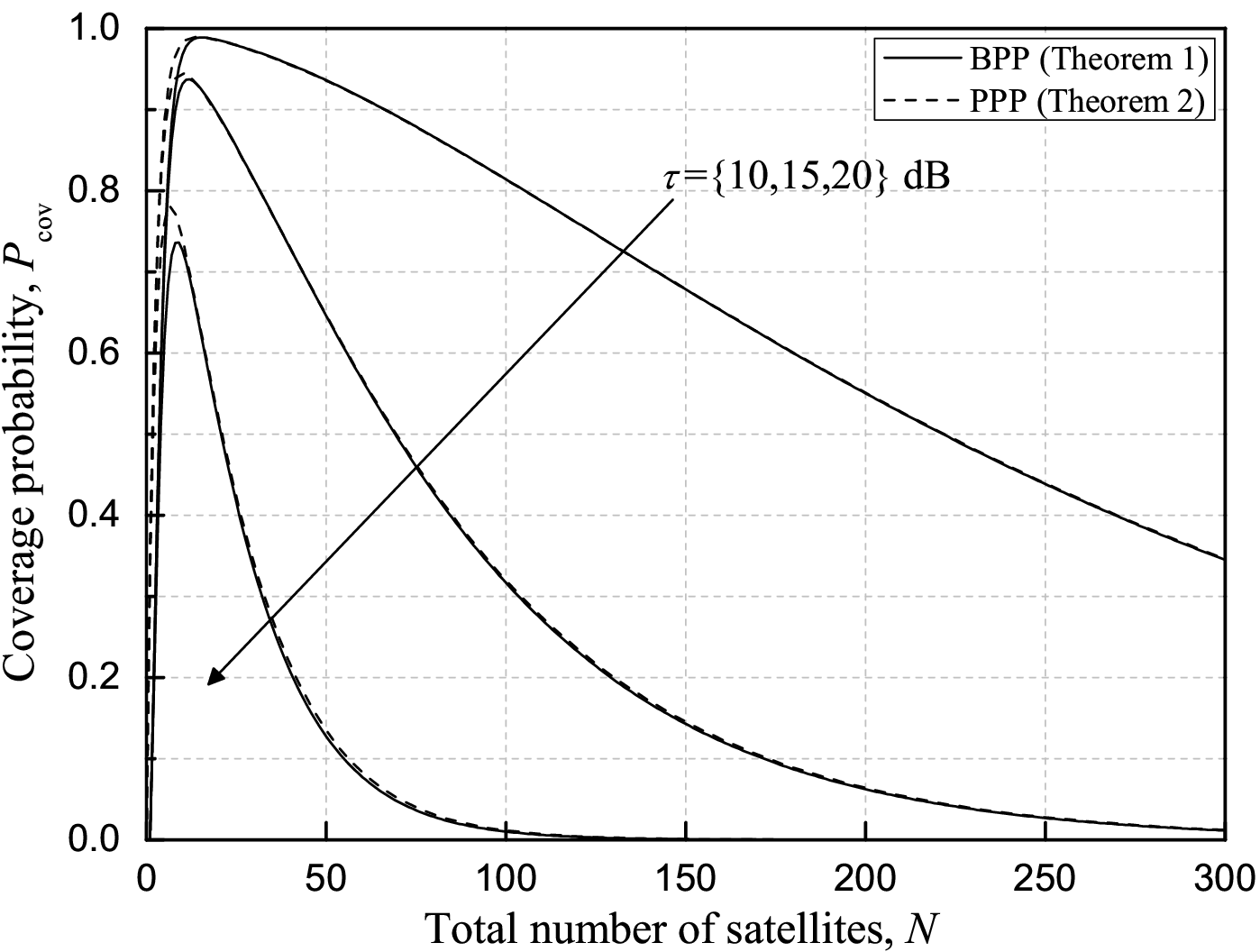}
\label{Fig:Pcov_vs_N}
}
\caption{The coverage probabilities versus the SINR threshold $\tau$ and the number of GEO satellites $N$.}
\label{Fig:Pcov}
\end{figure*}

Fig. \ref{Fig:Pcov} shows the numerical results of coverage probabilities. 
The BPP-based analytical results are given from Theorem~\ref{thm:Pcov}, while the PPP-based analysis comes from Theorem \ref{thm:Pcov_PPP}. 
In Fig.~\ref{Fig:Pcov}~\subref{Fig:Pcov_vs_tau}, the coverage probability in Theorem \ref{thm:Pcov} provides a fairly close performance to the simulation results for various path loss exponents $\alpha=\{3,3.7\}$, verifying the effectiveness of the BPP-based modeling of the GEO distribution.
Fig.~\ref{Fig:Pcov}~\subref{Fig:Pcov_vs_N} shows the coverage probability versus the number of GEO satellites $N$. As $N$ increases, the coverage probability first increases until $N$ reaches a certain value, and then decreases. When the number of GEO satellites is small, deploying additional satellites enhances system performance by increasing the satellite visibility and the received signal-to-noise ratio from the serving satellite.
However, for a large number of satellites, the presence of more satellites can lead to increased interference, resulting in a degradation of coverage performance. 
The approximated coverage probability in Theorem 2 is fairly similar to the one in Theorem 1, especially for high $N$, which verifies the Poisson limit theorem-based approximation.

Fig. \ref{Fig:Pcov_vs_phi} shows the coverage probability versus the latitude of the terminal. As we expected, the coverage probability highly relies on the terminal's latitude. This is because the geostationary orbit is on the equatorial plane, resulting in unequal satellite visibility from different latitudes. This phenomenon successfully explains the fact that GEO satellites cannot provide any coverage to polar regions due to their inherent orbital characteristics. Furthermore, unlike the case with a relatively small number of satellites, e.g., $N=10$, when there are a large number of satellites, e.g., $N=100$ or $200$, the regions with high latitudes, e.g., $|\phi|>60$ degrees, have the higher coverage performance compared to those near the equator. This is mainly because when many GEO satellites interfere with each other, the region that sees a shorter visible arc of the geostationary orbit has better performance due to less interference. Despite this fact, when the sidelobe of the satellites' beam patterns, the main factor causing interference, is designed to be sufficiently low, e.g., $\Go/\Gn=30$ dB, the coverage performance around all latitudes is enhanced, and simultaneously the performance gap between different latitudes decreases.
In addition, the terminals at the equator, i.e., $\phi=0$ degrees, achieve slightly higher coverage performance than those near the equator because of smaller path losses and better satellite visibility.

\begin{figure*}
\centering
\subfigure[$\Go/\Gn=20$ dB]{
\includegraphics[width=.98\columnwidth]{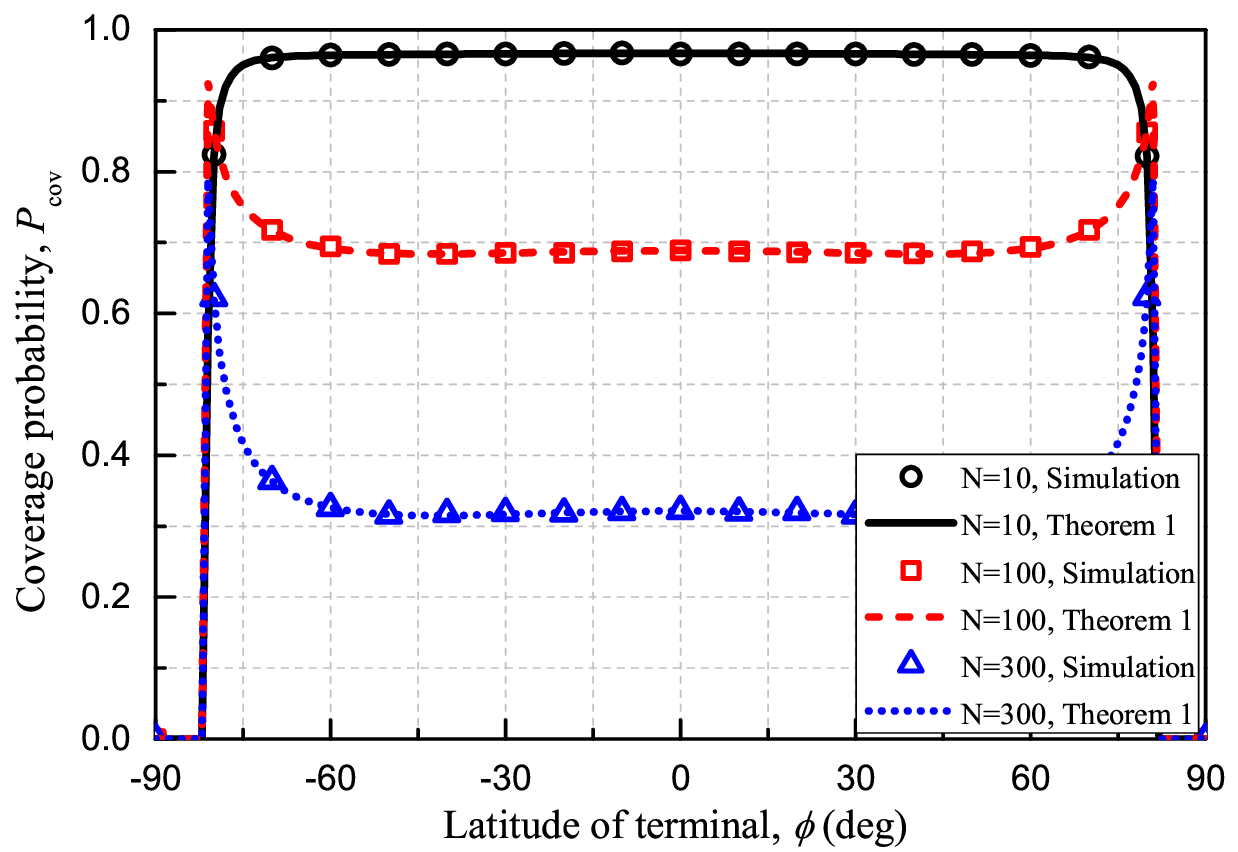}
\label{Fig:Pcov_vs_phi_gbar=20}
}
\subfigure[$\Go/\Gn=30$ dB]{
\includegraphics[width=.98\columnwidth]{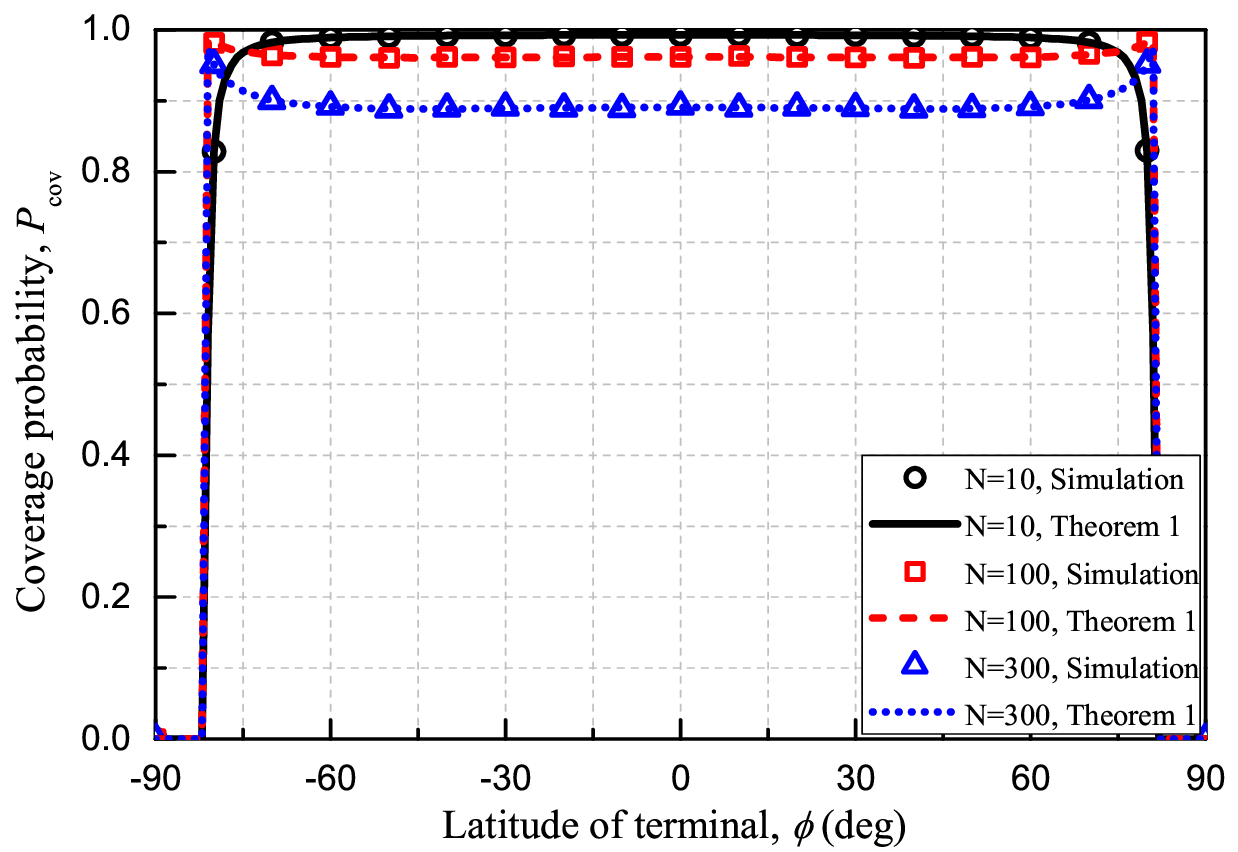}
\label{Fig:Pcov_vs_phi_gbar=30}
}
\caption{The coverage probability versus the terminal's latitude for various antenna gain ratios $\{\Go/\Gn\}=\{20,30\}$ dB with $m=1$ and $\tau=0$ dB.}
\label{Fig:Pcov_vs_phi}
\end{figure*}

\section{Conclusions}\label{sec:Conclusions}
In this paper, we investigated a novel approach to model the distribution of geosynchronous Earth orbit (GEO) satellites according to a binomial point process. We analyzed the distance distributions and the probabilities of distribution cases for the serving satellite. We also derived the coverage probability, and the approximated expression was obtained by using the Poisson limit theorem. Simulation results well matched the derived expressions, and the approximate performance was fairly close to the actual system performance. The impacts of the signal-to-interference-plus-noise ratio threshold, the number of GEO satellites, and the latitude of the terminal were discussed in terms of coverage probabilities. The analytical results are expected to give a fundamental framework for understanding GEO satellite networks and offer guidance when designing practical techniques for the heterogenous satellite communication systems.

\appendices

\section{Proof of Lemma \ref{lem:Pcases}}\label{App:Pcases}
    Using the finite-dimensional distribution of the BPP $\BPPg$, the probability that $q$ satellites are positioned in the visible arc $\Agvis$ is given by [\ref{Ref:Jung4}]
    \begin{align}\label{eq:fdd}
        \P&[\BPPg(\Agvis)=q, \: \BPPg(\Agvis^{\mathrm{c}})=N-q]
        =\binom{N}{q} \PsuccV^q (1-\PsuccV)^{N-q}
    \end{align}
    where 
    $\Agvis^{\mathrm{c}}$ is the invisible arc, i.e., the arc under the horizontal plane, whose length is $|\Agvis^{\mathrm{c}}|=|\Ag|-|\Agvis|$. 
    By substituting $q=0$ and $1$ in \eqref{eq:fdd}, we can obtain  $\P[\BPPg(\Agvis)=0]$ and
    $\P[\BPPg(\Agvis)=~1]$, respectively. The probability of Case 3 can be given by $\P[\BPPg(\Agvis)>1]=1-\P[\BPPg(\Agvis)=~0]-\P[\BPPg(\Agvis)=1]$, which completes the proof.

\section{Proof of Lemma \ref{lem:CDFR}}\label{App:CDFR}
    Let $D$ denote the distance from the terminal to an arbitrary satellite. Then, the probability that $D$ is less than or equal to $r$ is equivalent to the probability that the satellite is located within $\mathcal{A}(r)$, i.e., the success probability for $\mathcal{A}(r)$, which is given by
    \begin{align}\label{eq:CDFD}
        \P[D\le r]=\frac{|\mathcal{A}(r)|}{|\mathcal{A}|}=\Psi(r,\phi).
    \end{align}
    Since the satellites in $\BPPg$ are independent and identically distributed (i.i.d.), 
    the CDF of $R$ is given by 
    \begin{align}\label{eq:CDFR-1}
        F_R(r) 
        &= 1- \P[R > r] 
        \mathop=\limits^{(a)} 1- \left(1-\P[D \le r]\right)^{N}
    \end{align}
    where ($a$) follows from the independence of the distances to the satellites.
    The CDF is obtained by substituting \eqref{eq:CDFD} into \eqref{eq:CDFR-1}, and the PDF is derived by differentiating the CDF, which completes the proof.

\section{Proof of Lemma \ref{lem:CDFR0}}\label{App:CDFR0}
    Note that the maximum distance between the terminal and the visible satellite is defined as $\romax$. With this definition,
    the CDF of $R_0$ is given by
    \begin{align}\label{eq:CDFR0-1}
        F_{R_0}(r) &= \P\left[R\le r|\BPPg(\Agvis)>0\right]=\frac{\P[R \le r, R \le \romax]}{\P[R\le \romax]}.
    \end{align}
    Using the CDF of $R$ given in Lemma \ref{lem:CDFR}, \eqref{eq:CDFR0-1} is expressed as \eqref{eq:CDFR0}. 
    The PDF is directly obtained by differentiating \eqref{eq:CDFR0}, which completes the proof.

\section{Proof of Lemma \ref{lem:CDFRn}}\label{App:CDFRn}
    We now explore a specific case where given $R_0=r_0$, the distance to a satellite is larger than $r_0$ and less than or equal to $r$. 
    In this case, the satellite is positioned in $\Ag(r)\cap\Ag(r_0)^{\mathrm{c}}$ because the distance to the nearest satellite is already fixed to $R_0=r_0$. The probability of this case is interpreted as the success probability for the arc $\Ag(r)\cap\Ag(r_0)^{\mathrm{c}}$, which can be computed as the ratio of $|\Ag(r)\cap\Ag(r_0)^{\mathrm{c}}|$ to $|\Ag\cap\Ag(r_0)^{\mathrm{c}}|$, i.e.,
    \begin{align}\label{eq:prob_sc}
        \P[r_0 < D & \le r|R_0=r_0] = \frac{|\Ag(r)\cap\Ag(r_0)^{\mathrm{c}}|}{|\Ag\cap\Ag(r_0)^{\mathrm{c}}|} \nonumber\\
        &= \frac{|\Ag(r)|-|\Ag(r_0)|}{|\Ag|-|\Ag(r_0)|} = \frac{\Psi(r,\phi)-\Psi(r_0,\phi)}{1-\Psi(r_0,\phi)}.
    \end{align}

    For a given $R_0=r_0$, the CDF of $R_n$ is given by 
    \begin{align}\label{eq:CDFRn-1}
        F_{R_n|r_0}(r) 
        &= \P[R_n \le r|R_0=r_0]\nonumber\\
        &=\P[D \le r|R_0=r_0, r_0<D \le \romax]\nonumber\\
        &=\frac{\P[D \le r, r_0 < D \le \romax|R_0=r_0]}{\P[r_0<D \le \romax|R_0=r_0]}\nonumber\\
        &=
        \begin{cases} 
            0, & \mbox{if  } r<r_0,\\
             \frac{\P[r_0<D \le r|R_0=r_0]}{\P[r_0<D \le \romax|R_0=r_0]}, & \mbox{if  } r_0 \le r < \romax,\\
            1, & \mbox{otherwise.}
        \end{cases}
    \end{align}
    From \eqref{eq:prob_sc} and \eqref{eq:CDFRn-1}, we can obtain the CDF of $R_n$ given $R_0 = r_0$ as in \eqref{eq:CDFRn}. The PDF is directly obtained by using the derivative of $\Psi(r,\phi)$, given by $ \frac{d \Psi(r,\phi)}{d r} 
        =  \frac{2 r/\pi}{\sqrt{v_1-\left(v_2-r^2\right)^2}}$.
    This completes the proof.

\section{Proof of Lemma \ref{lem:LI}}\label{App:LTG}
 The Laplace transform is derived as 
    \begin{align}\label{eq:LI-1}
    &\mathcal{L}_{I|r_0}(s)
        = \E_{\BPPg,\{h_n\}}\left[\exp\left(-s\sum_{n=1}^{\BPPg(\Agvis\cap\Ag(r_0)^{\mathrm{c}})} \Pt \Gtn \Gr h_n \ell(\mathbf{x}_n)\right)\right]\nonumber\\
        &= \E_{\BPPg,\{h_n\}}\left[\prod_{n=1}^{\BPPg(\Agvis\cap\Ag(r_0)^{\mathrm{c}})}{\exp\left(-s \Pt \Gtn \Gr h_n \ell(\mathbf{x}_n)\right)}\right]\nonumber\\
        &\mathop=\limits^{(a)} \E_{\BPPg,\{R_n\}}\left[\prod_{n=1}^{\BPPg(\Agvis\cap\Ag(r_0)^{\mathrm{c}})}{\mathcal{L}_{h_n}\left( \frac{s}{\un R_n^{\alpha}} \right)}\right]\nonumber\\
        &\mathop=\limits^{(b)} \E_{\BPPg}\left[\prod_{n=1}^{\BPPg(\Agvis\cap\Ag(r_0)^{\mathrm{c}})} \int_{r_0}^{\romax}{\mathcal{L}_{h_n}\left( \frac{s}{\un r_n^{\alpha}} \right)} f_{R_n|r_0}(r_n) dr_n\right]\nonumber\\
        &\mathop=\limits^{(c)} \E_{\BPPg}\left[\prod_{n=1}^{\BPPg(\Agvis\cap\Ag(r_0)^{\mathrm{c}})} \int_{r_0}^{\romax}\left(\frac{m \un r_n^{\alpha}}{s+m \un r_n^{\alpha}}\right)^m f_{R_n|r_0}(r_n) dr_n\right]\nonumber\\
        &\mathop=\limits^{(d)}  
        \sum_{\nI=0}^{N-1} \P\left[\BPPg\left(\Agvis\cap\Ag(r_0)^{\mathrm{c}}\right)=\nI\right] \nonumber\\
        &\quad\times   \prod_{n=1}^{\nI}\int_{r_0}^{\romax}\left(\frac{m \un r_n^{\alpha}}{s+m \un r_n^{\alpha}}\right)^m f_{R_n|r_0}(r_n) dr_n    
    \end{align}
    where ($a$) follows from the i.i.d. distribution of the channel gains $h_n$, ($b$) follows from the i.i.d. distribution of the distances $R_n$, ($c$) follows because $\mathcal{L}_{h_n}(s)=\left(\frac{m}{s+m}\right)^m$, and ($d$) follows from the law of total expectation, i.e., $\E[X]=\sum_i\P[A_i]\E[X|A_i]$. 
    According to Lemma \ref{lem:num_int_sat}, the probability $\P\left[\BPPg\left(\Agvis\cap\Ag(r_0)^{\mathrm{c}}\right)=\nI\right]$ in \eqref{eq:LI-1} is derived as $\binom{N-1}{\nI}\PsuccI^{\nI}\PsuccI^{N-1-\nI}$, which completes the proof.

\section{Proof of Theorem \ref{thm:Pcov}}\label{App:Pcov}
Considering the satellite visible probability identified in Section \ref{sec:visibility_anal}, the coverage probability is given by
    \begin{align}
    \Pcov 
    &=\P[\BPPg(\Agvis)=0]\,\P[\SINR\geq\tau|\BPPg(\Agvis)=0]\nonumber\\
    &\quad + \P[\BPPg(\Agvis)>0]\,\P[\SINR\geq\tau|\BPPg(\Agvis)>0]\nonumber\\
    &\mathop=\limits^{(a)}\P[\BPPg(\Agvis)>0]\,\P[\SINR\geq\tau|\BPPg(\Agvis)>0]
    \end{align}
    where ($a$) follows because when there is no visible satellite, i.e., $\BPPg(\Agvis)=0$, the corresponding coverage probability $\P[\SINR\geq\tau|\BPPg(\Agvis)=0]$ is zero.
    When $\BPPg(\Agvis) > 0$, the coverage probability is derived as
    \begin{align}\label{eq:Pcovt-1}
    &\P[\SINR\geq\tau|\BPPg(\Agvis) > 0]\nonumber\\
        &= \E_{R_0}\left[ \P\left[\frac{\Pt \Gto \Gr h_0 \ell(\mathbf{x}_0)}{N_0 W + I} \geq \tau \,\bigg|\, R_0=r \right]\right]\nonumber\\
        &= \int_{\rmin}^{\romax} \P\left[h_0 \geq \uo (I + N_0 W) \tau r^{\alpha}  \,|\, R_0=r \right]f_{R_0}(r)dr\nonumber\\
        &= \int_{\rmin}^{\romax} \E_I\left[\P\left[h_0 \geq \uo (I + N_0 W) \tau r^{\alpha}  \,|\, R_0=r, I \right]\right]f_{R_0}(r)dr\nonumber\\
        &\mathop \approx \limits^{(a)} \int_{\rmin}^{\romax} \E_I\left[\sum_{i=1}^{m}\binom{m}{i}(-1)^{i+1}e^{-\nu i \uo (I + N_0 W) \tau r^{\alpha}} \; \right]f_{R_0}(r)dr\nonumber\\
        &\!=\! \sum_{i=1}^{m}\!\binom{m}{i}(-1)^{i+1}\!\!\int_{\rmin}^{\romax}  \!\!e^{-\nu i \uo N_0 W \tau r^{\alpha}} \E_I[e^{-\nu i \uo I \tau r^{\alpha}} \;]f_{R_0}(r)dr
    \end{align}
    where ($a$) follows from the approximated CDF of the channel gain $F_{h_n}(x)=1-\sum_{i=1}^{m}\binom{m}{i}(-1)^{i+1}e^{-\nu i x}$ [\ref{Ref:Andrews}].
    From the definition of the Laplace transform, i.e., $\mathcal{L}_X(s) = \E_X[{e^{-sX}}]$, we obtain \eqref{eq:Pcov}, which completes the proof.

\section{Proof of Lemma \ref{lem:LI_approx}}\label{App:LI_approx_proof}
The Laplace transform is derived as
\begin{align}\label{eq:LI_approx_proof}
    &\tilde{\mathcal{L}}_{I|r_0}(s)
        = \E\left[e^{-sI}|R_0 = r_0 \right]\nonumber\\
        &\mathop=\limits^{(a)} \E_{\BPPg}\left[\prod_{n=1}^{\bar{\BPPg}(\Agvis\cap\Ag(r_0)^{\mathrm{c}})}{ \E_{h_n}\left[e^{-s \Pt \Gn h_n \ell(\mathbf{x}_n)}\right] }\right]\nonumber\\
        &\mathop=\limits^{(b)} \exp \biggl(-\lambda \!\!\int_{\mathbf{x}_n \in \bar{\BPPg}(\Avis \cap \Ag(r_0)^{\mathrm{c}})} \!\!\left( 1\!-\!\E_{h_n}\!\!\left[e^{-s \Pt \Gn h_n \ell(\mathbf{x}_n)}\right] \right) d\mathbf{x}_n \bigg) \nonumber\\
        &\mathop=\limits^{(c)} \exp \Biggl(-\lambda \!\!\int_{\mathbf{x}_n \in \bar{\BPPg}(\Avis \cap \Ag(r_0)^{\mathrm{c}})} \left( 1- \frac{1}{\left(\frac{s r^{-\alpha}}{m \omega}+1\right)^{m}} \right) d\mathbf{x}_n \Biggr) \nonumber\\
        &\mathop=\limits^{(d)} \exp\left(-\frac{2N}{\pi} \int_{r_0}^{\romax} \left( 1- \frac{1}{\left(\frac{s r^{-\alpha}}{m \omega}+1\right)^{m}} \right)  \frac{r dr}{\sqrt{v_1-\left(v_2-r^2\right)^2}} \right) \nonumber\\
        &\mathop= \exp\vast(-\frac{2N}{\pi} \underbrace{\int_{r_0}^{\romax}  \frac{r dr}{\sqrt{v_1-\left(v_2-r^2\right)^2}}}_{\Omega_1(\romax)-\Omega_1(r_0)} \nonumber\\
        &\quad\qquad+ \frac{2N}{\pi} \underbrace{\int_{r_0}^{\romax} \frac{1}{\left(\frac{s r^{-\alpha}}{m \omega}+1\right)^{m}}   \frac{r dr}{\sqrt{v_1-\left(v_2-r^2\right)^2}}}_{\Omega_2(s,r_0)} \vast) 
    \end{align}
    where ($a$) follows from the independence of channel gains $h_n$,
    ($b$) follows from the Campbell’s theorem for the PPP $\bar{\BPPg}$, ($c$) follows from the Laplace transform $\E[e^{-s h_n}]=\mathcal{L}_{h_n}(s)=\left(\frac{m}{s+m}\right)^m$, and ($d$) comes from 
    \begin{align}
        \frac{d|\Ag(r)|}{dr}
        &=\frac{4r(\re+\ag)}{\sqrt{v_1-\left(v_2-r^2\right)^2}}.
    \end{align}
     Using the derivative of $\Omega_1(r)$, $\frac{d}{dr}\Omega_1(r)=\frac{r}{\sqrt{v_1-\left(v_2-r^2\right)^2}}$, and the definition of $\Omega_2(s,r_0)$ 
     given in Lemma \ref{lem:LI_approx},
     the proof is complete.

\ifCLASSOPTIONcaptionsoff
  \newpage
\fi


\begin{thebibliography}{1}
\bibitem{bib:3GPP_38.811}\label{Ref:3GPP_38.811}
3GPP TR 38.811 v15.4.0, ``Study on NR to support non-terrestrial networks," Sep. 2020.

\bibitem{bib:3GPP_38.821}\label{Ref:3GPP_38.821}
3GPP TR 38.821 v16.0.0, ``Solutions for NR to support non-terrestrial networks (NTN)," Dec. 2019.



\bibitem{bib:Su-22}\label{Ref:Su-22}
Y. Su, Y. Liu, Y. Zhou, J. Yuan, H. Cao, and J. Shi, ``Broadband LEO satellite communications: Architectures and key technologies," \emph{IEEE Wireless Commun.}, vol. 26, no. 2, pp. 55-61, Apr. 2019.

\bibitem{bib:Okati1}\label{Ref:Okati1}
N. Okati, T. Riihonen, D. Korpi, I. Angervuori, and R. Wichman, ``Downlink coverage and rate analysis of low Earth orbit satellite constellations using stochastic geometry," \emph{IEEE Trans. Commun.}, vol. 68, no. 8, pp. 5120-5134, Aug. 2020.

\bibitem{bib:Talgat1}\label{Ref:Talgat1}
A. Talgat, M. A. Kishk, and M.-S. Alouini, ``Nearest neighbor and contact distance distribution for binomial point process on spherical surfaces," \emph{IEEE Commun. Lett.}, vol. 24, no. 12, pp. 2659-2663, Dec. 2020.

\bibitem{bib:Talgat2}\label{Ref:Talgat2}
——, ``Stochastic geometry-based analysis of LEO satellite communication systems," \emph{IEEE Commun. Lett.}, vol. 25, no. 8, pp. 2458-2462, Aug. 2021.

\bibitem{bib:Jung1}\label{Ref:Jung1}
D.-H. Jung, G. Im, J.-G. Ryu, S. Park, H. Yu, and J. Choi, ``Satellite clustering for non-terrestrial networks: Concept, architectures, and applications," \emph{IEEE Veh. Technol. Mag.}, vol. 18, no. 3, pp. 29-37, Sep. 2023.


\bibitem{bib:Jung2}\label{Ref:Jung2}
D.-H. Jung, J.-G. Ryu, W.-J. Byun, and J. Choi, ``Performance analysis of satellite communication system under the shadowed-Rician fading: A stochastic geometry approach," \emph{IEEE Trans. Commun.}, vol. 70, no. 4, pp. 2707-2721, Apr. 2022.

\bibitem{bib:Al-Hourani1}\label{Ref:Al-Hourani1}
A. Al-Hourani, ``An analytic approach for modeling the coverage performance of dense satellite networks," \emph{IEEE Wireless Commun. Lett.}, vol. 10, no. 4, pp. 897-901, Apr. 2021.

\bibitem{bib:Al-Hourani2}\label{Ref:Al-Hourani2}
——, ``Optimal satellite constellation altitude for maximal coverage," \emph{IEEE Wireless Commun. Lett.}, vol. 10, no. 7, pp. 1444-1448, July 2021.

\bibitem{bib:Okati2}\label{Ref:Okati2}
N. Okati and T. Riihonen, ``Nonhomogeneous stochastic geometry analysis of massive LEO communication constellations," \emph{IEEE Trans. Commun.}, vol. 70, no. 3, pp. 1848-1860, Mar. 2022.

\bibitem{bib:Park}\label{Ref:Park}
J. Park, J. Choi, and N. Lee, ``A tractable approach to coverage analysis in downlink satellite networks," \emph{IEEE Trans. Wireless Commun.}, vol. 22, no. 2, pp. 793-807, Feb. 2023.

\bibitem{bib:Abdu}\label{Ref:Abdu}
T. S. Abdu, S. Kisseleff, E. Lagunas, S. Chatzinotas, and B. Ottersten, ``Joint carrier allocation and precoding optimization for interference-limited GEO satellite," in \emph{Proc. 39th International Communications Satellite Systems Conference (ICSSC 2022)}, 2022, pp. 128-132.

\bibitem{bib:Abdu2}\label{Ref:Abdu2} 
T. S. Abdu, S. Kisseleff, E. Lagunas, and S. Chatzinotas, ``Flexible resource optimization for GEO multibeam satellite communication system," \emph{IEEE Trans. Wireless Commun.}, vol. 20, no. 12, pp. 7888-7902, Dec. 2021.

\bibitem{bib:Jia}\label{Ref:Jia}  
M. Jia, X. Liu, X. Gu, and Q. Guo, ``Joint cooperative spectrum sensing and channel selection optimization for satellite communication systems based on cognitive radio," \emph{Int. J. Satell. Commun. Network.}, vol. 35, no. 2, pp. 139-150, Dec. 2015.


\bibitem{bib:Khan}\label{Ref:Khan} 
W. U. Khan, E. Lagunas, A. Mahmood, B. M. ElHalawany, S. Chatzinotas, and B. Ottersten, ``When RIS meets GEO satellite communications: A new sustainable optimization framework in 6G," in \emph{Proc. IEEE 95th Veh. Technol. Conference: (VTC-Spring)}, Jun. 2022, pp. 1-6.


\bibitem{bib:CS.Park}\label{Ref:CS.Park}
C.-S. Park, C.-G. Kang, Y.-S. Choi, and C.-H. Oh, ``Interference analysis of geostationary satellite networks in the presence of moving non-geostationary satellites," in \emph{Proc. 2nd Int. Conf. Inf. Technol. Converg. Services}, 2010, pp. 1-5.

\bibitem{bib:Fortes}\label{Ref:Fortes} 
J. M. P. Fortes, R. Sampaio-Neto, and J. E. A. Maldonado, ``An analytical method for assessing interference in interference environments involving NGSO satellite networks," \emph{Int. J. Satell. Commun.}, vol 17., no. 6, pp. 399-419, Dec. 1999.

\bibitem{bib:SGP4}\label{Ref:SGP4} 
D. Vallado and P. Crawford, ``SGP4 orbit determination," in \emph{Proc. AIAA/AAS Astrodynamics Specialist Conference and Exhibit}, Aug. 2008, pp. 6770-6799.

\bibitem{bib:Jung3}\label{Ref:Jung3}
D.-H. Jung, J.-G. Ryu, and J. Choi, ``Satellite clusters flying in formation: Orbital configuration-dependent performance analyses," doi: arXiv:2305.01955.

\bibitem{bib:Chiu}\label{Ref:Book:Chiu}
S. N. Chiu, D. Stoyan, W. S. Kendall, and J. Mecke, \emph{Stochastic Geometry and Its Applications}, 3nd ed. New York, NY: Wiley, 2013.

\bibitem{bib:Jung4}\label{Ref:Jung4}
D.-H. Jung, J.-G. Ryu, and J. Choi, ``When satellites work as eavesdroppers," \emph{IEEE Trans. Inf. Forensics Security}, vol. 17, pp. 2784-2799, 2022.

\bibitem{bib:Andrews}\label{Ref:Andrews}
J. G. Andrews, T. Bai, M. N. Kulkarni, A. Alkhateeb, A. K. Gupta, and R. W. Heath, ``Modeling and analyzing millimeter wave cellular systems," \emph{IEEE Trans. Commun.}, vol. 65, no. 1, pp. 403-430, Jan. 2017.

\end{thebibliography}
\end{document}